\documentclass[sigconf,authorversion,nonacm,10pt]{acmart}
\usepackage[english]{babel}
\usepackage[utf8]{inputenc}
\usepackage{graphicx}
\graphicspath{ {./Figs/} }
\usepackage{textcomp}
\usepackage{xcolor}
\usepackage{subfigure}
\usepackage{balance}
\usepackage[ruled, lined, linesnumbered, commentsnumbered, noend]{algorithm2e}
\usepackage[noend]{algorithmic}
\usepackage{thmtools} 
\usepackage{thm-restate}
\usepackage{hyperxmp,hyperref}
\hypersetup{
 colorlinks,
 citecolor=NavyBlue,
 linkcolor=ForestGreen,
}
\newtheorem{theorem}{Theorem}

\newtheorem{definition}{Definition}

\newcommand{\etal}{\textit{et al.}}
\newcommand{\Delay}{\textsc{Delay}\xspace}
\newcommand{\Greedy}{\textsc{Greedy}\xspace}
\newcommand{\MIP}{\textsc{MIP}\xspace}
\newcommand{\NoDo}{\renewcommand\algorithmicdo{}}

\AtBeginDocument{%
  }

\setcopyright{acmcopyright}
\acmYear{2022}\copyrightyear{2022}
\acmConference[SOSR 22]{Symposium on SDN Research}{October 19--20, 2022}{Virtual Event, USA}
\acmBooktitle{Symposium on SDN Research (SOSR 22), October 19--20, 2022, Virtual Event, USA}
\acmPrice{15.00}
\acmDOI{10.1145/3563647.3563655}
\acmISBN{978-1-4503-9892-3/22/10}

\begin{document}

\title[The Augmentation-Speed Tradeoff for Consistent Network Updates]{The Augmentation-Speed Tradeoff \\ for Consistent Network Updates}

\author{Monika Henzinger}
\email{monika.henzinger@univie.ac.at}
\orcid{0000-0002-5008-6530}
\affiliation{
  \institution{Department of Computer Science, 
  \\ University of Vienna}
  \city{Vienna}
  \country{Austria}
}
\author{Ami Paz}
\email{ami.paz@LISN.fr}
\orcid{0000-0002-6629-8335}
\affiliation{
  \institution{LISN - CNRS
  \\ \& Paris-Saclay University}
  \city{Paris}
  \country{France}
}
\author{Arash Pourdamghani}
\email{pourdamghani@tu-berlin.de}
\orcid{0000-0002-9213-1512}
\affiliation{
  \institution{TU Berlin}
  \city{Berlin}
  \country{Germany}
}
\author{Stefan Schmid }
\email{stefan.schmid@tu-berlin.de}
\orcid{0000-0002-7798-1711}
\affiliation{
  \institution{TU Berlin \& Fraunhofer SIT}
  \city{Berlin}
  \country{Germany}
}
\begin{abstract}
Emerging software-defined networking technologies enable more adaptive communication infrastructures, allowing for quick reactions to changes in networking requirements by exploiting the workload’s temporal structure. However, operating networks adaptively is algorithmically challenging, as meeting networks’ stringent dependability requirements relies on maintaining basic consistency and performance properties, such as loop freedom and congestion minimization, even during the update process.
This paper leverages an augmentation-speed tradeoff to significantly speed up consistent network updates. 
We show that allowing for a small and short (hence practically tolerable, e.g., using buffering) oversubscription of links allows us to solve many network update instances much faster, as well as to reduce computational complexities (i.e., the running times of the algorithms). We first explore this tradeoff formally, revealing the computational complexity of scheduling updates. We then present and analyze algorithms that maintain logical and performance properties during the update. Using an extensive simulation study, we find that the tradeoff is even more favorable in practice than our analytical bounds suggest. In particular, we find that by allowing just 10\% augmentation, update times reduce by more than 32\% on average, across a spectrum of real-world networks.
\end{abstract}

\sloppy



\maketitle

\section{Introduction}

To render communication networks more dependable,
the networking community currently makes great efforts
to automate network operations. The envisioned ''self-driving'' 
networks~\cite{FeamsterR18} can relieve operators of their most complex tasks, hence minimizing the chances for human errors, which frequently are 
the cause of major outages~\cite{BeckettMMPW16,chirgwin2017google,ems1_2018}. 
Furthermore, more automated networks enable more adaptive
network operations, allowing to quickly react to network events
such as shifts in the demand and hence to exploit temporal
structure in the traffic patterns for optimizations~\cite{AvinGG020,RoyZBPS15,pieee22}. 
These more automated and adaptive network operations are 
enabled, among others, by emerging software-defined 
and programmable networking technologies, that allow direct control over the forwarding tables of switches and routers.

A programmatic and software-defined control and update of forwarding paths can be attractive in many situations~\cite{FoersterSV19}.
For example, fast route updates can be useful in reacting to security policy changes or security threats, by actively rerouting traffic through a firewall. In wide-area networks, Internet Service Providers may adjust their traffic engineering policy in reaction to changes in the load. 
Adaptions to the routes taken by packets may also be required to react to link failures or to support maintenance work or service relocations.

However, a more adaptive network operation introduces an algorithmic challenge:
in order to meet the stringent dependability and performance requirements, networks need to be reconfigured \emph{quickly and consistently}. 
A key challenge here is that updates at different switches occur asynchronously, and update times can vary significantly, between milliseconds to fractions of a second~\cite{JinLGKMZRW14, UpToSeconds}.
Especially when adaptions are frequent, it is important that the network fulfill certain properties, such as congestion freedom and loop freedom, \emph{even during the update}.

Over the last few years, the consistent network update problem has received much attention in the literature. 
Seminal work focused on logical properties~\cite{ReitblattFRSW12,LudwigMS15,ForsterMW16,MahajanW13,mcclurg2015efficient,ludwig2014good}, but also performance aspects received much attention early on~\cite{FoersterSV19,BrandtFW16,liu2013zupdate,christensen2021latte}. While the question of how to reroute flows in a congestion-free manner is still not well-understood algorithmically (especially if one requires that algorithms come with provable performance and approximation guarantees), it has been shown recently that the resulting update schedules can be long and complex, requiring many rounds of updates~\cite{AmiriDSW18,LudwigMS15}.

This paper is motivated by the observation that already a small and short oversubscription of links (which we will refer to as augmentation) can lead to significantly faster update schedules (see Figure~\ref{fig: Intro} for an illustration). This, in turn, may allow more fine-grained network operations.
Such augmentation is often feasible in practice and mitigated by buffering and congestion control if the augmentation is bounded in magnitude and time. Short oversubscriptions are common today in congestion control, especially unproblematic in virtual networks, which provide soft capacity constraints, and typically do not affect prices~\cite{laoutaris2009delay}.

\begin{figure}[t]
    \centering
    \includegraphics[width=0.5\textwidth, trim=10 0 0 0, clip]{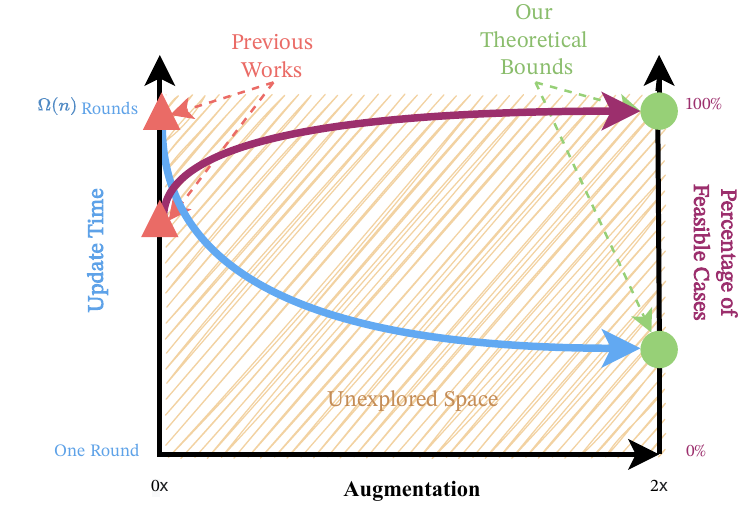}
    \caption{This paper explores the benefits of augmentation on the speed and feasibility of network updates. While prior work (red triangles on the left) did not consider augmentation (the dashed orange area), our approach with augmentation provides flexibility for operators and hence supports more fine-grained network operations. The green circles represent our theoretically optimal bounds that we proved in this paper, and curves represent the tradeoff that we saw in our empirical results, as can also be seen in the counterpart of this qualitative plot in the evaluation section (Figure~\ref{fig: IntroData}).}
     \label{fig: Intro}  
\end{figure}

\subsection{Our Contributions}
This paper argues that existing literature on network update scheduling ignores the fact that short link oversubscriptions are unproblematic due to buffering, and uncovers an interesting tradeoff between the tolerable oversubscription and the speed at which networks can be updated, measured by the number of rounds in the rerouting schedule.
In general, several challenges might occur when trying to update a network's routing policy:
infeasibility, i.e., cases where an update schedule without overloading any link simply does not exist; 
speed, i.e., the time it takes for the updates to complete;
and computability, i.e., cases where an update schedule exists, but finding a feasible or optimal update schedule is NP-hard.
We show that all these challenges can be overcome by allowing a small oversubscription of the communication links.
In this paper, 
we will distinguish between additive and multiplicative \emph{augmentation} of a link.
Additive augmentation refers to the maximum capacity increase on any given link in the network, and multiplicative augmentation is a factor by which we multiply capacities.

We first explore this tradeoff analytically and show that under a factor-2 multiplicative augmentation, fast update schedules is always feasible; we also show that for smaller factors, the problem of computing update schedules is NP-hard---i.e., a short schedule may exist, but finding it is computationally infeasible in the worst case.
We then present both optimal and fast algorithms to exploit the augmentation-speed tradeoff while provably maintaining basic consistency properties. 
We report on an extensive simulation study using real-world networks, 
and we find that we can reach higher speeds with a slight augmentation.
More precisely, a 10\%  augmentation can reduce update times by 32\% on average across a range of networks derived from the Internet Topology Zoo.
As a contribution to the research community, we release our experimental artifacts as well as our simulation code (as open source) together with this paper at \hyperref[GitHub]{github.com/inet-tub/AugmentRoute}.

\subsection{Organization}
The remainder of this paper is organized as follows.
In~\S\ref{sec: model}, we introduce our formal model and define the properties which need to be maintained transiently. \S\ref{sec:theory} details the analytical study of the valid update schedules under
augmentation and derives hardness results. 
Then we present two polynomial-time greedy algorithms and an optimal algorithm based on mixed integer programming in~\S\ref{sec:algorithms}, and explore their behavior on real-world networks in~\S\ref{sec:empricalResults}.
After reviewing the related work in~\S\ref{sec:relatedWork}, we conclude our contribution in~\S\ref{sec:conclusion}.

\section{Modelling Consistent Network Updates and Tradeoffs}\label{sec: model}

We model a network as a directed graph $G=(V,E)$. The set $V$ consists of $n$ nodes 
representing the switches in the network, and the set $E \subseteq V \times V$ of $m$ directed edges denoting the links.
A directed edge $e=(v,w) \in E$ connects the \emph{tail} node $v$ to the \emph{head} node $w$. An edge $e$ has a real-valued positive capacity $c_e \in \mathbb{R}_{\geq0}$, and $C_{\max} = \max_e c_e$ is the maximum edge capacity among all edges.

Flows are unsplittable and routed along unique paths, which are dictated by the network's routing policy.
When the policy changes, it may become necessary to update the routing path by updating the outgoing edges of the nodes (i.e., the forwarding rules).
We consider flow pairs consisting of an \emph{old} flow (which is already in use)
and an \emph{updated} flow  (which the new policy enforces).
We emphasize that flows in each flow pair share the same source and terminal node. We also consider that the forwarding is done based on \emph{both} the source and the terminal. 
\begin{definition}[Flow pairs]
A set of $k$ \emph{flow pairs} is defined as $P = \{P_1,\dots,P_k\}$, where each flow pair $P_i$ consists of an old flow  $F_i^o$ and an updated flow $F_i^u$. Flows of the $i^{th}$  flow pair are \emph{unsplittable}, which means each is only a simple $s_i$-$t_i$ path in $G$. Each flow corresponds to a real-valued positive demand $d_i$, initiating from the same source node $s_i$ and ending at the same terminal node $t_i$.
\end{definition}

In order to update a flow pair $P_i$ from and old flow $F_i^o$ to an updated flow $F_i^u$, we need to update all the nodes that are appearing only in $F_i^o$ or only in $F_i^u$.
An \emph{update schedule} from the old flow $F_i^o$ to the updated flow $F_i^u$ is a sequence of \emph{update rounds}.
In each round, a subset of the nodes changes their outgoing edges from the edges used in $F_i^o$ to the edges used in $F_i^u$.
Formally, for a flow pair $P_i$, we define an $R$-round 
update schedule $U_i = \{U_i^1,\dots, U_i^R\}$
such that in each round $r= 1,\dots,R$, a set of nodes that were not included in previous updates, i.e.,  $U_i^r\subseteq V(F_i^o \cup F_i^u)\setminus (U_i^1\cup\cdots\cup U_i^r-1)$, are updated.
We assume that all the changes in the same round happen asynchronously, i.e., in an unpredictable order. 
This makes the problem harder, as a worst-case order inside each update set must be taken into consideration.

To maintain consistency during network updates, an update schedule must provide \emph{loop freedom} and \emph{congestion freedom}.

\begin{figure}[t]
    \centering
    \subfigure[Initial flow]
    {
    \centering
        \includegraphics[scale=1.1]{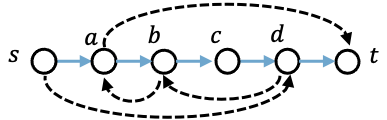}
    }
    \subfigure[After updating nodes $a$ and $d$]
    {
        \centering
        \includegraphics[scale=1.1]{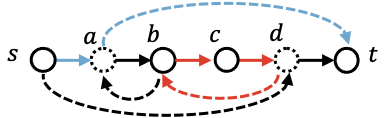}
    }
    \caption{Solid lines represent the old flow, dashed lines the updated flow. Solid circles show nodes that are not updated, and dotted circles are updated nodes. (a)~Initially, a flow passes through blue lines. (b)~After the update $a$ and $d$, a transient loop appears (red lines) that violates loop-freedom property,  even though the terminal is still reachable (through blue lines).}
     \label{fig: loop}  
\end{figure}

\subsection{Loop Freedom}

When scheduling batches of updates simultaneously, individual updates at nodes can happen at different times, which might cause transient \emph{forwarding loops}, see Figure~\ref{fig: loop} for an example.
\emph{Loop freedom} requires that forwarding loops never happen during the update of a flow pair, regardless of whether nodes in the loop can be reached from the source node or not.

In the case of nodes that are in the old path but not in the updated path, $F_i^o\setminus F_i^u$, their routing policy is updated from edge to no-edge. Since we need to ensure reachability at all times, we update these nodes after the nodes leading to them.
Similarly, nodes in $F_i^u\setminus F_i^o$ initially have not been assigned an outgoing edge and need to be updated before the nodes leading to them. 
Note that together with loop-freedom, handling these cases guarantees that a path from the source to the terminal node exists at all times:
Each node that has an incoming edge has a single outgoing edge (except for the terminal node), and there are no loops in the graph.

\subsection{Congestion Freedom}

Assume that for each flow pair, we have an update schedule that is valid and loop-free.
We want to make sure that we can apply all the updates simultaneously without causing congestion on the network links.
For this, fix an update round $r$ in all flow pairs, and
for each flow pair $i$, consider the temporary flow $F_i$, 
that includes edges from  old flow $F^o_i$ and edges from the updated flow $F^u_i$ that pass flow "during" round $r$, i.e. edges from old flow that has not been updated "before" round $r$, edges from the updated flow that will be used "after" round $r$.

\begin{definition}[Valid schedule]
An update schedule is \emph{valid} if at any time,
the set of temporary flows $F = \{F_1,\dots,F_m\}$
satisfies that for every edge $e \in E$,
the sum of demands of flows that pass through an edge is at most its capacity, i.e.  $\forall e \in E: \sum_{j:e \in F_j} d_j\leq c_e$.
\end{definition}

This paper is motivated by the benefits of slight augmentation of the current capacity.  We investigate two possibilities for augmentation: \emph{multiplicative augmentation} and \emph{additive augmentation}. In the first approach, we consider all capacities multiplied by a real number $\alpha \geq 1$. In the latter one, we allow capacities to be increased by a fixed number $\beta \geq 0$.

\begin{definition}[$(\times\alpha)$-valid schedule, $(+\beta)$-valid schedule]
	For $\alpha\geq1, \beta\geq0$, 
	we say that a schedule is \emph{$(\alpha,\beta)$-valid} if at any time, 
	the set of temporary flows $F = \{F_1,\dots,F_m\}$
	satisfies that for every edge $e \in E$,
	we have
	$\sum_{j:e \in F_j} d_j\leq\alpha c_e+\beta$.
	A schedule is \emph{$(\times\alpha)$-valid} if it is $(\alpha,0)$-valid,
	and \emph{$(+\beta)$-valid} if it is $(1,\beta)$-valid.
\end{definition}

We first show that every update schedule is $(\times2)$-valid and $(+C_{\max})$-valid, and then prove that for any $\epsilon > 0$, , the two problems of deciding if a $(\times(2-\epsilon))$-valid update schedule exists, and if a $(+(C_{\max}/3-\epsilon))$-valid update schedule exists, are both NP-hard.

\section{Theoretical Analysis}\label{sec:theory}
In this section, we start exploring the speed-congestion tradeoff analytically. 
We first derive upper bounds on the required additive and multiplicative augmentation that make update schedules valid. 
We further explore the computational complexity of finding valid update schedules, showing it is NP-hard to decide whether a valid update schedule exists when the augmentation is below some threshold.

\subsection{Upper Bounds}
The following theorem characterizes the amount of multiplicative and additive augmentation needed to render update schedules valid.
\begin{theorem}
	\label{thm:2 is enough}
	Every update schedule is $(\times 2)$-valid and $(+C_{\max})$-valid.
\end{theorem}

\begin{proof}
	Consider an update schedule at any given point in time during update, and an edge $e$.
	Let $S^o$ be the set of indices of flows that are not yet updated at this time point at $e$, and $S^u$ the set indices of flows that are updated;
	the current load on $e$ is at most $\sum_{i\in S^o}d_i+\sum_{i\in S^u}d_i$.
	
	Since the set of old flows is valid, we have $\sum_{i:e\in F^o_i}d_i\leq c_e$, and similarly the set of updated flows is valid and 
	$\sum_{i:e\in F^u_i}d_i\leq c_e$.
	Note that $S^o\subseteq \{i:e\in F^o_i\}$ and $S^u\subseteq \{i:e\in F^u_i\}$, and hence
	$\sum_{i\in S^o}d_i\leq \sum_{i:e\in F^o_i}d_i\leq c_e$ and
	$\sum_{i\in S^u}d_i\leq \sum_{i:e\in F^u_i}d_i\leq c_e$.
	Thus, 
	$\sum_{i\in S^o}d_i+\sum_{i\in S^u}d_i\leq 2c_e$, as desired.	
	
	For the additive case, note that $2c_e\leq c_e+C_{\max}$, and the proof immediately follows from the multiplicative case.
\end{proof}
These bounds leave the question of how to find a schedule that minimizes the number of rounds,
a question that needs to be resolved for each flow pair separately.
We explore techniques to minimize the number of rounds later in the paper.

\subsection{Hardness Results}\label{sec:lowerBound}

We next study the computational complexity and prove a tight converse of the multiplicative case in Theorem~\ref{thm:2 is enough},
and a non-tight converse for the additive case.

\begin{restatable}{theorem}{multi}
	\label{thm:2-eps is nph}
	For every constant $1/3 > \epsilon>0$, deciding if
	a
	$(\times(2-\epsilon))$-valid update schedule exists is NP-hard.
\end{restatable}

The proof of this theorem extends a construction from~\cite{AmiriDP0W19}, and shows that an algorithm for finding a $(\times(2-\epsilon))$-valid update schedule can be used in order to find a satisfying assignment for a 3-CNF formula of the well-known 3SAT problem~\cite{3SAT}. In the 3SAT problem, we should assign 0-1 values to a set of binary variables such that a set of boolean clauses become valid. Each clause consists of only three variables, or their negation.

The details of the proof can be found in Appendix~\ref{appendix: omitted}. Our proof transforms each clause or variable into a series of \emph{gadgets}. A gadget is a series of old and updated paths between a pair of nodes, designed to transform a valid update into a solution for 3SAT.

\begin{restatable}{theorem}{adi}
	\label{thm:additive is nph}
	For every constant $1 > \epsilon>0$ deciding if a $(+(C_{\max}/3-\epsilon))$-valid update schedule exists is NP-hard.
\end{restatable}

The proof of the theorem of the additive case goes along the lines of the proof of the multiplicative case of Theorem~\ref{thm:2-eps is nph}. 
The proof of the additive case introduces a new variable gadget between pairs of source-terminal nodes of the valid update problem.
We present details of the proof in the Appendix~\ref{appendix: omitted} for completeness of the paper. 

\section{Algorithms}\label{sec:algorithms}

This section presents algorithms that navigate and exploit the tradeoff between speed and augmentation. We first detail a fast algorithm that will be useful for efficient updates in large networks. In addition, we provide an optimal algorithm based on mixed integer programming, which is useful for small networks and will also serve as a baseline in the evaluation of fast algorithms.

\subsection{Fast Algorithm for Short Update Schedules}
Our first approach is a natural greedy algorithm that computes a short loop-free update schedule.
We then extend this algorithm by adding a post-processing step that reduces augmentation by allowing delays in update schedules. These algorithms are fast enough to be used in most of real-world scenarios.

\paragraph{Algorithm \Greedy}
For this algorithm, we assume a fixed flow pair $i$ and a certain round.
As mentioned in the preliminaries, if a node $v$ is only in the updated flow, it should be updated before nodes that are in both updated and old flows, and if $v$ is only in the old flow, it should be updated after these nodes; \Greedy{} starts by doing exactly that.
Thus, we only need to describe how to order the nodes that belong both to the old and the updated flow of the flow pair $i$.

For each flow pair, we maintain a set $A_i$ of  \emph{active} edges, edges that are actively passing the flow for the pair $i$. 
Also, define $\overline{A_i}$ as the edges from $F_i^u$ that are not yet updated,
and sort them based on their distance to the terminal node in the updated flow $F_i^u$.
Before the first round, $A_i$ consists of all the edges of the old flow, and $\overline{A_i}$ includes all the edges from the updated flow.
During each round, we go through the edges of  $\overline{A_i}$, starting from the one nearest to the terminal node,
and for each edge, if it does not create a directed cycle in the graph induced by $A_i$, we add the edge to $A_i$, and remove it from $\overline{A_i}$, then mark its tail to be updated in the current round.
Such an edge always exists in each update round:
In the first round, the edge of $F_i^u$
nearest to the terminal is directly connected to the terminal, and thus can be updated without creating a cycle.
Later, $F_i^u\cap A_i$ is a set of disjoint paths, and one of them leads to the terminal node. 
On this path, consider the node $v$ that is furthest away from the terminal node;
then, $F_i^u$ contains an edge $(u,v)\in \overline{A_i}$, and node $u$ can be updated without creating a cycle (otherwise, the set $\overline{A_i}$ is empty, which means that update has been finished successfully).

After going through all of $\overline{A_i}$, we update $A_i$ by removing all the edges from $F_i^o$ whose tail was updated, and move to the next round.
Note that in each round, at least the edge currently closest to the terminal will be added to $A_i$ and removed from $\overline{A_i}$; thus, the algorithm ends after at most $|F_i^u|$ rounds.

\begin{figure}[t]
    \centering
    \includegraphics[scale=1.1]{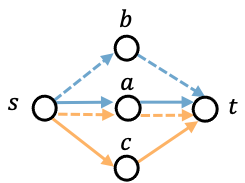}
    \caption{Example where delay reduces the congestion.
    The first flow pair consists of $F^o_1=(s,a,t),F^u_1=(s,b,t)$ and the second is $F^o_2=(s,c,t),F^u_2=(s,a,t)$;
    all demands and capacities equal to $1$. 
    During the second round of the \Greedy{} algorithm, it is possible that node $s$ gets updated for the second flow pair before it is updated for the first flow pair,
    in which case the edges $(s,a)$ and $(a,t)$ get congested. 
    This is avoided by the delay algorithm, which shifts the schedule of the second flow pair by one round.}
     \label{fig: Delay Helps}  
\end{figure}

\paragraph{Algorithm \Delay}
We propose an improved algorithm, called \Delay{}, that modifies any valid schedule by delaying selected flow pairs up to $T$ rounds (usually only 1-2 rounds) to reduce augmentation while increasing the number of rounds up to $T$.
\Delay{} operates in phases. In each phase, the algorithm greedily chooses the flow pair $P_i$ and the delay value $d_i \le T$ 
such that delaying the start of $P_i$'s update schedule by $d_i$ rounds provides the highest decrease the augmentation. \Delay terminates if no such flow pairs exist anymore. Note that the schedule of a flow pair can be delayed in multiple phases.

Figure~\ref{fig: Delay Helps} shows the potential benefit of delaying the update of flow pairs. In this example, delaying the bottom flow pair for one round eliminates congestion in the network.

\subsection{Optimal Algorithm}
We next present an optimal algorithm that is based on a \emph{Mixed Integer Program} (\MIP) for finding a loop-free update schedule that minimizes either the number of rounds or the augmentation.
For the sake of explanation, we consider a fixed multiplicative (additive) augmentation ratio $\alpha$ ($\beta$) and describe the minimization of the number of rounds $R$.
However, given a fixed number of rounds, with a simple change to the \MIP, we can minimize the augmentation instead\footnote{One can think of optimizing both at the same time using biconvex optimization.}.

\paragraph*{Variables:}
Let us fix the flow pair $i$ and the round $r$. We then describe the set of variables related to a node $v \in V(F_i^o \cup F_i^u)$. 
\begin{itemize}
    \item \emph{Node update variable} $x^r_{v,i}$ is a binary variable that indicates whether node $v$ from flow pair $i$ updates in round $r$ or not.
\end{itemize}
We define special variables for nodes that appear in both old and updated flows.
A node $v \in V(F_i^o \cap F_i^u)$ is a \emph{branching} node if it has an outgoing edge $(v,w)$ in the old flow $F^o_i$ and another outgoing edge $(v,w')$ in the updated flow $F_i^u$.
Similarly, $v \in V(F_i^o \cap F_i^u)$ is a \emph{merging} node if it has an incoming edge $(w,v)$ in the old flow $F^o_i$ and another incoming edge $(w',v)$ in the updated flow $F_i^u$.

After each branching node $v$, there exists a merging node by, denoted by $v'$. 
The node $v'$ must exit since the old and updated flows need to intersect again, as they share a common terminal node.
Furthermore, no other branching node appears on the old and updated flows between $v$ and $v'$, as a branching node needs to be part of both old and updated flows. Hence a branching node $v$ is uniquely matched with a merging node $v'$.

\begin{itemize}
    \item \emph{Branching variable} $\Lambda^r_{v,i}$ is a binary variable indicating whether node $v$ is a branching node in the flow pair $i$ and gets updated in round $r$.
    \item \emph{Merging variable} $\Upsilon^r_{v,i}$ is a binary variable indicating whether node $v$ is a merging node in the flow pair $i$ and receives flow from both old and updated flows during round $r$.
\end{itemize}

To avoid creating loops during an update round $r$, we assign an order to all nodes in the flow pair $i$.

\begin{itemize}
    \item \emph{Ordering variable} $o^r_{v,i}$ is an integer variable assigned to node $v$ of flow pair $i$ in round $r$ in the range of $[1,n]$.
\end{itemize}

Now we look at the variables for an edge $e = (v,w) \in F_i^o \cup F_i^u$ of the flow pair $i$. 
We say edge $e$ is \emph{active} during (after) round $r$ if it appears in $F^o_i$ and its tail $v$ has not been updated, or it is part of $F^u_i$ and $v$ has been updated. If an edge $e$ is active for pair $i$, it might be part of the source-terminal flow of pair $i$.

\begin{itemize}
    \item \emph{Edge activity variable} $y^r_{e,i}$ is a binary variable indicating if edge $e$ is active \emph{after} round $r$ in flow pair $i$ or not.
    \item \emph{Edge transitivity variable} $\gamma^r_{e,i}$ is a binary variable indicating whether edge $e$ is active \emph{during} round $r$ in the flow pair $i$ or not.
    \item \emph{Edge flow variable} $f^r_{e,i}$ is fractional variable in the range $(0,1)$ indicating whether the source-terminal flow of flow pair $i$ passes over $e$ during round $r$ or not.
\end{itemize}

\NoCaptionOfAlgo
\begin{algorithm}
\caption{Mixed Integer Program to Compute Optimal Solution}
\small
 \begin{algorithmic}[1]
 \label{line:objective}
 \NoDo
 \STATE \textbf{Minimize} $R$ (or $\alpha$, $\beta$)
 \FORALL{$i \in [|P|]$}
    \STATE $\sum_{r \in [R]} x_{v,i}^r = 1$ \hfill  $\forall  v \in V(F_i^o \cup F_i^u) \setminus \{t_i\}$
    \label{line:UpdateOnce}
    \STATE $y_{(v,w),i}^0 = 1$ \hfill $\forall (v,w) \in F^o_i$
    \label{Line:OldFlow}
    \STATE $y_{(v,w),i}^0 = 0$ \hfill $\forall (v,w) \notin F^o_i$
    \label{Line:UpdatedFlow}
    \FORALL{$r \in [R]$}
    \label{line:forRounds}
      \STATE $R \geq r\cdot x_{v,i}^r$   \hfill  $\forall  v \in V(F_i^o \cup F_i^u) \setminus \{t_i\}$
      \label{line:MaxRoute}
    \STATE $y_{(v,w),i}^r = 1$ \hfill  $\forall (v,w) \in F^o_i \cap F^u_i$
     \label{line:remainActive}
      \STATE $y_{(v,w),i}^r = \sum_{r' \leq r} x_{v,i}^{r'}$ \hfill $\forall (v,w) \in F^u_i \setminus F^o_i$
  \label{line:update-edge-update}
 \STATE $y_{(v,w),i}^r = 1-\sum_{r' \leq r} x_{v,i}^{r'}$ \hfill $\forall (v,w) \in F^o_i \setminus F^u_i$
   \label{line:old-edge-update}
    \FORALL{$\forall (v,w) \in F^o_i \cup F^u_i $}
        \STATE $ \gamma_{(v,w),i}^r \geq y_{(v,w),i}^{r-1}$ 
        \label{line:Trans1}
        \STATE $ \gamma_{(v,w),i}^r \geq y_{(v,w),i}^r $ 
        \label{line:Trans2}
        \STATE $ \gamma_{(v,w),i}^r \leq \frac{o^r_{w,i} - o^r_{v,i} - 1}{|V|-1} + 1$ 
        \label{line:level}
    \ENDFOR
    \FORALL{$\forall v \in P_i $}
        \STATE $\Lambda^r_{v,i} = x^r_{v,i}$ \hfill  $\exists (v,w) \in F^o_i \land (v,w') \in F^u_i$
        \label{line:isFork}
        \STATE $\Lambda^r_{v,i} = 0$ \hfill  $\nexists (v,w) \in F^o_i \land (v,w') \in F^u_i$
        \label{line:notFork}
        \STATE $\Upsilon^r_{v,i} \leq f^r_{(w,v),i},f^r_{(w',v),i} $ \hfill  $\exists (w,v) \in F^o_i \land (w',v) \in F^u_i$
        \label{line:isJoin}
        \STATE $\Upsilon^r_{v,i} = 0 $ \hfill  $\nexists (w,v) \in F^o_i \land (w',v) \in F^u_i$
        \label{line:notJoin}
    \ENDFOR
    \STATE $f^r_{(v,w),i} \leq \gamma^{r}_{(v,w),i} $ \hfill  $\forall (v,w) \in F^o_i \cup F^u_i $
    \label{line:ValidFlows}
    \STATE $\sum_{(s_i,v)} f^r_{(s_i,v),i} = 1 + \Lambda^r_{s_i,i}$ \hfill  $s_i \in P_i$
    \label{line:FlowSource}
    \STATE $\sum_{(v,t_i)} f^r_{(v,t_i),i} = 1 + \Upsilon^r_{t_i,i}$ \hfill  $t_i \in P_i$
    \label{line:FlowTerminal}
    \STATE $\sum_{(v,w)} f^r_{(v,w),i} - \sum_{(w',v)} f^r_{(w',v),i} = \Lambda^r_{v,i} - \Upsilon^r_{v,i} $ 
    \\  $\forall  v \in v \in V(F_i^o \cup F_i^u) \setminus \{s_i,t_i\}$ 
    \\ $(v,w),(w',v) \in F^o_i \cup F^u_i$
    \label{line:flow-all}
  \ENDFOR
 \STATE $\sum_{i \in [|U|]} f_{(v,w),i}^r \cdot d_i  \leq \alpha \cdot c_{(v,w)} + \beta$  \hfill $\forall (v,w) \in E$
  \label{line:CapacityCheck}
 \ENDFOR
\end{algorithmic}
\end{algorithm}

\paragraph*{Constraints:}
As before, let us fix a flow pair $i$.
We now go through properties that an optimal schedule needs to satisfy. 

\paragraph{Node update}
Each node in the flow pair $i$ except for the terminal needs to be updated exactly once.
Thus, for any given node $v \in V(F_i^o \cup F_i^u)$, the node update variable $x^r_{v,i}$ should be equal to $1$ in precisely one of the rounds, and $0$ in others. Hence, the sum of $x^r_{v,i}$ over all rounds should be $1$, as shown in Constraint~\ref{line:UpdateOnce}.
Also, the number of rounds is lower bounded by the last round a node updates, implying Constraint~\ref{line:MaxRoute}.

During round $r$, node $v$ is a branching node if and only if it is updating from an edge $(v,w)$ in the old flow to a different edge $(v,w')$ in the updated flow of flow pair $i$ (see Constraints~\ref{line:isFork} and~\ref{line:notFork}).
Based on our definition, a node is a merging node in round $r$ if it receives flow from two different edges, as shown in Constraints~\ref{line:isJoin} and~\ref{line:notJoin}.

\paragraph{Edge activity} 
Constraints~\ref{Line:OldFlow} and~\ref{Line:UpdatedFlow} guarantee that only edges that are active in the initial round, round $0$, are edges from the old flow.

After the round $0$, if an edge was part of both old and updated flow, i.e. $e \in F_i^o \cup F_i^u$, edge $e$ remains active (Constraint~\ref{line:remainActive}.
Otherwise, if an edge $e$ is only in $F_i^o$, it remains active \emph{until} the round in which its tail $v$ is updated (Constraint~\ref{line:old-edge-update}), and if the edge $e$ is only in $F_i^u$ it becomes active \emph{after} node $v$ updates (Constraint~\ref{line:update-edge-update}).

If the activity of an edge $e$ changes during a round $r$, it can happen at any time during the round $r$. 
Therefore we say edge $e$ is active \emph{during} round $r$ if it was active in the previous round (Constraint~\ref{line:Trans1}) or it becomes active after this round (Constraint~\ref{line:Trans2}).

\paragraph{Loop freedom}
The loop freedom property does not allow any transitive loops to appear at any time during the update.
Thus, for all edges $e = (v,w) \in F_i^o \cup F_i^u$ that can be active during around, i.e. $\gamma_{(v,w),i}^r =1$, we need the order of edge's head, the node $w$, be always larger than edge's tail, the node $v$.
Inspired by the Miller-Tucker-Zemlin constraint used in sub-tour elimination in traveling salesman problem~\cite{miller1960integer}, we define Constraint~\ref{line:level} that prevents loops from being created. Note that the right side of the inequality assumes a value in $[1,2)$ if $o^r_{v,i} < o^r_{w,i}$ and a value less than $1$, otherwise.

\paragraph{Congestion freedom} 
Similar to loop-freedom, we need to maintain congestion freedom at any time during the update.
Any flow pair $i$ during round $r$ can only use the edges that are active during that round (see Constraint~\ref{line:ValidFlows}).

As flows are unsplittable, if flow pair $i$ uses edge $e$, it passes all of its demand $d_i$ through it. 
That is why we can define $f^r_{e,i}$ as a binary value and just multiply them by $d_i$ in Constraint~\ref{line:CapacityCheck}, which ensures the augmented capacity of edge $e$ limits the sum of the flows that pass through $e$.

We know that for any given branching node $v$ in flow pair $i$, there exists a corresponding merging node $u$. If node $v$ updates in round $r$, all edges on both paths between $v$ and $u$ must be active during round $r$. Hence, edges in both paths are prone to congestion.
That is why in round $r$, from branching node $v$, we send flow on both paths, and in merging node $u$, we unify these flows.
By Constraint~\ref{line:FlowSource}, we enforce the source node of flow pair $i$ to send a unit of flow. If the source node is also a branching node, it can temporarily send one additional unit of flow in the round that it updates.
Similarly, the terminal node of flow pair $i$ needs to receive one unit flow, unless it is a merging node that receives one additional unit of flow (Constraint~\ref{line:FlowTerminal}).
For all other nodes of the flow pair $i$, the incoming flow must match the outgoing flow, unless the node is a branching node that can output an additional unit flow, or whether it is a merging node that decreases the flow by one unit (Constraint~\ref{line:flow-all}).

\paragraph*{Objectives:}
We describe two objectives, minimizing the number of rounds given a fixed augmentation and minimizing augmentation given a fixed number of rounds.

\paragraph*{Minimizing number of rounds}
To enforce loop freedom on each of the flow pairs, we might need up to $n-1$ rounds, $n = |U|$ is the number of nodes in the network~\cite{FoersterLMS18}.
As we saw in the discussion about the delay algorithm, in order to minimize the augmentation, it might be beneficial for us not to update each flow pair in each round. Thus, we consider the possibility that flow pairs update one after the each other, and consider the number of rounds $R$ to be in the range  $\{1,\dots,k\cdot(n-1)\}$, where $k$ is the number of flow pairs.
To minimize the number of rounds, we assume that we are given the allowed augmentation, i.e., both $\alpha$ and $\beta$, and run the mixed integer program with the objective to minimize~$R$.

\paragraph*{Minimizing Augmentation}
By Theorem~\ref{thm:2 is enough}, multiplicative augmentation factor $\alpha$ is at most $2$. Additive augmentation value $\beta$ can be between $0$ and the maximum capacity of an edge.
To minimize the multiplicative or additive augmentation, we assume that we are given the number of rounds $R$.
We change Constraint~\ref{line:forRounds} to iterate only up to $R$ rounds, and also remove Constraint~\ref{line:MaxRoute}. 
Finally, we change the objective in Line~\ref{line:objective} to minimize $\alpha$ (or $\beta$) instead of $R$. 
 
\section{Empirical Results}\label{sec:empricalResults}

We complement our analytical results by studying the augmentation-speed tradeoff in practical scenarios and performing an extensive simulation study. We first report our methodology and then present our main insights. The implementation is available at \hyperref[GitHub]{github.com/inet-tub/AugmentRoute}.

\subsection{Methodology}
We implemented our algorithms in python~3.6, using NetworkX~2.5~\cite{networkx}, Numpy~1.19~\cite{numpy} and Matplotlib~3.3~\cite{Matplotlib} libraries.
To solve the mixed integer program, we used Gurobi~9.1~\cite{gurobi}.
We have executed our program on a machine that has Intel Xeons E5-2697V3 SR1XF as CPU and provides 128 GB DDR4 RAM.

\paragraph{Topologies} We have evaluated our algorithms on $218$ connected and directed graphs from the Internet Topology Zoo~\cite{TopologyZoo} that have at most $100$ nodes. We removed trees from the graph set, in which any schedule can update trivially in one round. The limitation on the number of nodes is due to the high running time of the mixed integer program.
As the data set only provides the network topology, we now describe how to set the link capacities, flow pairs, and their demands. 

\paragraph{Generating Flows} 
For the flow pair $i$, we choose two distinct nodes uniformly at random as the source $s_i$ and terminal $t_i$.
If we use the shortest paths, old and updated paths fully overlap with each other, so instead, we extend the shortest path routing with Valiant routing~\cite{LevPV81} (or more generally, segment routing~\cite{FilsfilsNPCF15}) to \emph{segment path routing}. 
In the segment path routing, from nodes other than source and terminal, we randomly choose one of the waypoint nodes, $w^1_i$, and find the shortest $(s_i,w^1_i)$ and $(w^1_i,t_i)$ paths.
The basis of calculating shortest paths is the weights that we assign to each edge, in the range of $(1,100)$. 
Since the high running time of \MIP is a limiting factor, we use $250$ flow pairs in our experiments. 

\paragraph{Setting Link Capacities}
We construct a set of \emph{baseline flows} used to set the link capacities, with demand chosen uniformly at random from $(10,20)$. This method is preferable to other methods, such as setting the capacity of each edge independently, since it assigns link capacities proportionally to a possible link usage from old and updated paths.
 
\paragraph{Setting Demands of Old and Updated Flows} 
To generate demand for old and updated paths such that both of them remain valid,
we use an approach reminiscent of congestion control protocols (and in particular, the slow start algorithm~\cite{SlowStart,Jacobson88}).
We also set the initial demands of all old and updated flows to $1$, and arranged them in round-robin order.
We then multiply the demand by a growth factor $g$. We stop multiplying the demand of a flow pair if either the old or the updated flow cannot be increased due to link capacity constraints.
The value of the growth factor directly affects link utilization, as links get less utilized with an increased growth factor. In our experiments, the default value of the growth factor is $1.1$, as it provides more than 99\% link utilization.

\paragraph{Running \MIP}
We first run \Greedy and~\Delay on each input graph, receiving their update schedules. 
We then calculate the augmentation and the number of rounds that those schedules need.
In the end, we run the mixed integer program first based on the number of rounds that each of those algorithms needs, and also based on the augmentation that they require.

\begin{figure}[t]
    \centering
    \includegraphics[width=0.27\textwidth,trim=0 6 0 10, clip]{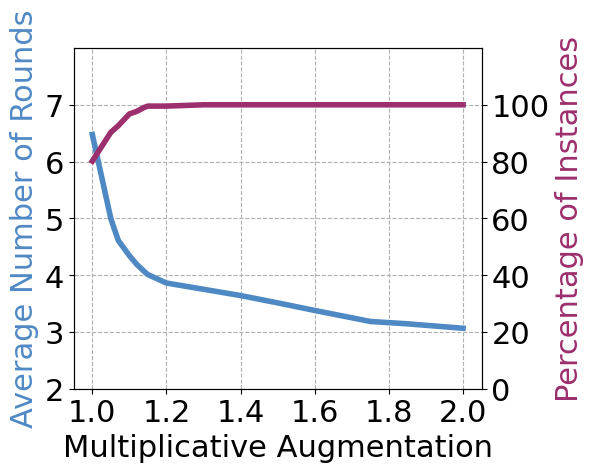}
    \caption{The experimental augmentation-speed tradeoff (blue) and 
    augmentation-feasibility tradeoff (purple).
    With only 20\% augmentation, we observe a sharp decrease from more than $6$ rounds to less than $4$ rounds on average.
    With as little as 15\% augmentation, the percentage of solvable flow pairs increases from 80\% to over 99\%.
    }
    \label{fig: IntroData}  
\end{figure}

\subsection{Results}
We first study the achievable improvement of optimal update schedules with additional augmentation, then evaluate the performance of our greedy algorithms, and finally describe how our algorithms can be used in other practical settings. 

\noindent\textbf{Benefits of Augmentation.}
To evaluate the benefits of augmentation, we use our optimal algorithm to find the best update schedule given a certain amount of augmentation. We only focus on multiplicative augmentation, as additive augmentation follows a similar trend.
Having Theorem~\ref{thm:2 is enough}, it is enough to check augmentations between $1$ and $2$.
We compare instances with different augmentations on two metrics: the average number of rounds which represents the speed of an update schedule, and the percentage of feasible solutions.
As shown in Figure~\ref{fig: IntroData},  the number of rounds drops considerably with a slight increase in augmentation. With only 5\%
additional augmentation, the number of rounds reduces by more than 22\%, and with 10\% augmentation, the number of rounds drops by more than 32\%.
The number of feasible solutions increases by more than 16\% with 10\% augmentation and grows above 99\% while using less than 15\% augmentation.

\noindent\textbf{Comparing Optimal and Fast Algorithms.}
We start by comparing the algorithms when the optimal algorithm is allowed the same augmentation as the \Greedy or \Delay algorithms, with the delay threshold equal to three. 
In Figure~\ref{fig: Rounds} we compare the percentage of optimal and \Greedy (\Delay) schedules that have a certain number of rounds. We see that the \MIP schedules more than half of the cases in three rounds, and the \Greedy and \Delay algorithms can have up to $6$ and $7$ rounds, respectively. The difference between the number of rounds matches our intuition since the cases that \Greedy algorithm performs poorly are rare, but they are possible.

\begin{figure}[t]
    \centering
    \subfigure[]
    {
    \centering
        \includegraphics[width=0.225\textwidth]{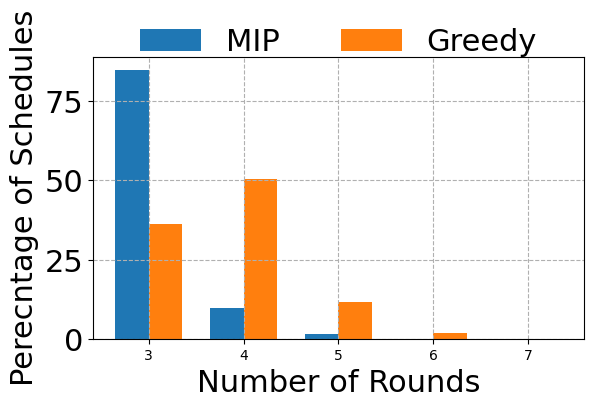}
    }
    \subfigure[]
    {
    \centering
        \includegraphics[width=0.225\textwidth]{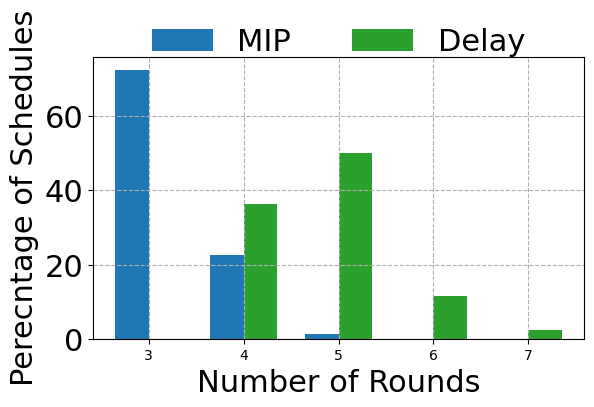}
    }
    \caption{
   The percentage of cases in which a given number of rounds is required for the optimal algorithm, compared to (a)~\Greedy and (b)~\Delay algorithms. 
   The optimal algorithm optimizes the number of rounds, given the multiplicative augmentation resulting from the other algorithm.}
     \label{fig: Rounds}  
\end{figure}

\begin{figure}[h]
    \centering
    \subfigure[]
    {
    \centering
        \includegraphics[width=0.226\textwidth]{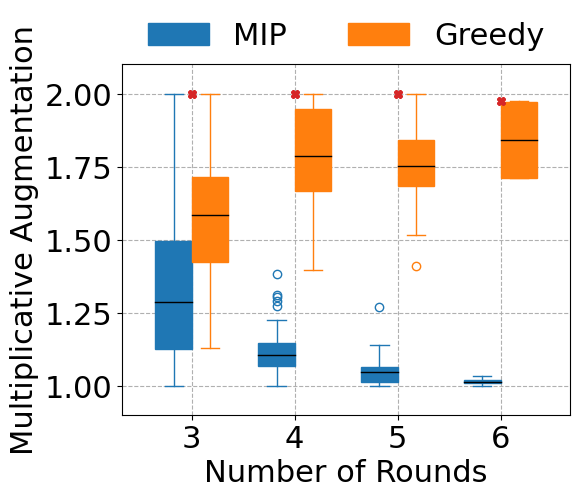}
    }
    \subfigure[]
    {
    \centering
        \includegraphics[width=0.226\textwidth]{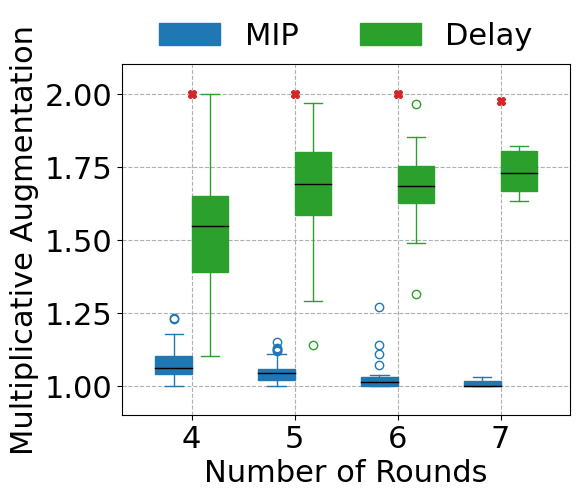}
    }
    \subfigure[]
    {
    \centering
        \includegraphics[width=0.226\textwidth]{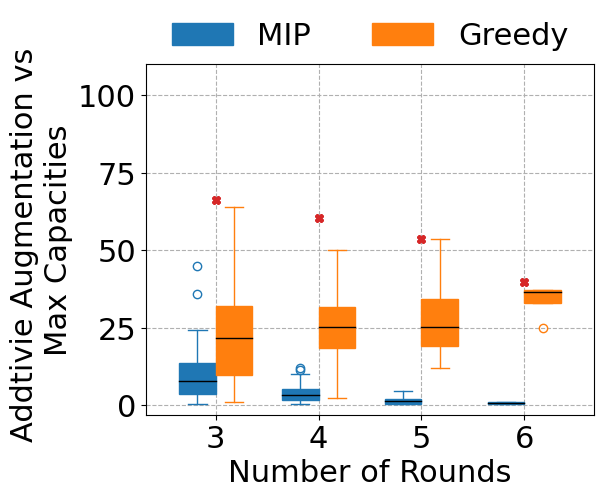}
    }
    \subfigure[]
    {
    \centering
        \includegraphics[width=0.226\textwidth]{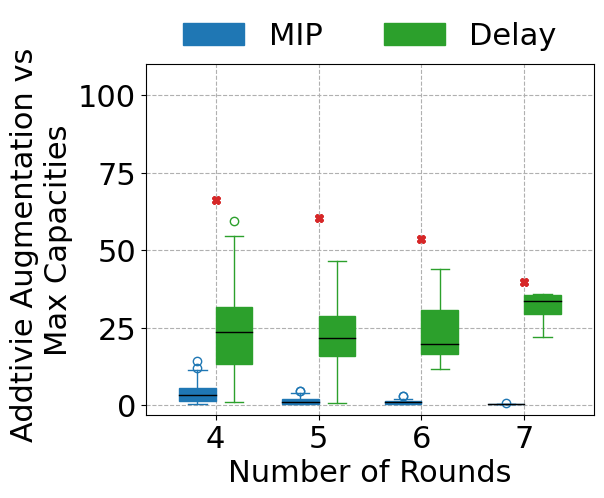}
    }
    \caption{
    Comparison of an optimal scheduling vs.~\Greedy (a,c) and \Delay (b,d), in terms of    
    multiplicative augmentation (a,b) and additive augmentation as percentage of $C_{\max}$ (c,d).
    In each case, the \MIP uses the number of rounds needed by the other algorithm.
    The red dots denote the worst-case scenario when all the possible old and updated flows overlap on each edge.}
     \label{fig: Augmentation}  
\end{figure}

We then fix the number of rounds from the \Greedy and \Delay algorithms and find the required augmentation for the mixed integer program. In Figure~\ref{fig: Augmentation}, we compare the average augmentation given fixed rounds. 
The top figures show the comparison between multiplicative augmentations as shown on the y-axis, and in the bottom figures, we can see additive augmentations.
When the number of rounds is limited, the optimal algorithm needs as high augmentation as \Greedy (\Delay), but when the number of rounds increases, the optimal algorithm can delay or provide a gap between updates of each flow pair, which leads to lower augmentation.

We compare the running time of our algorithms in Figure~\ref{fig: Time}. As expected from an integer program, the running time grows fast when increasing the number of nodes or flows (the variance is due to the various heuristics applied by the solver). The fast growth of running time of \MIP continues even after $100$ nodes. 
Hence, \Greedy and \Delay are particularly useful in large networks as well as when the main focus is on speed. 
On the other hand, the \MIP can be attractive to provide minimal augmentation on small networks.

\noindent\textbf{Takeaway for other settings.}
Given the above results, and based on today's flexible networks, allowing for momentary link augmentations can solve issues, such as long update schedules.
Particularly, our recommendation is that in events of long update schedules in practice, optimizing schedules based on efficient algorithms like \Greedy or \Delay could be enough to reduce the number of rounds significantly.

\begin{figure}[t]
    \centering
    \subfigure[]
    {
    \centering
   	\includegraphics[width=0.225\textwidth,trim=5 10 0 5, clip]{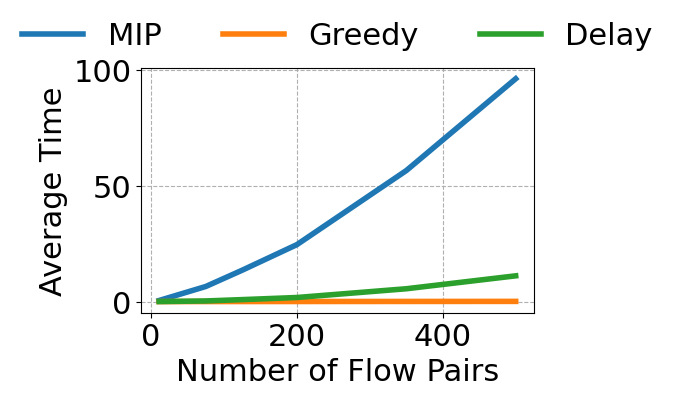}
    }
    \subfigure[]
    {
    \centering
    	\includegraphics[width=0.225\textwidth,trim=5 10 0 5, clip]{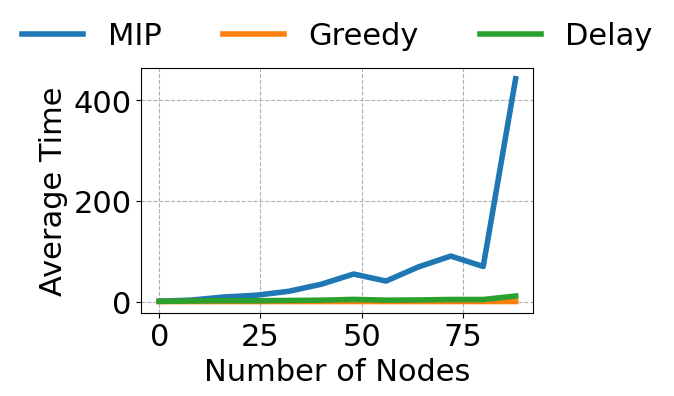}
    }
    \caption{
    Average time (in seconds) needed to run our algorithms on networks from the Internet Topology Zoo:
    (a)~on graphs with at most 50 nodes and varying numbers of flow pairs;
    (b)~on graphs of varying sizes and 250 flow pairs.}
     \label{fig: Time}  
\end{figure}

\section{Additional Related work}\label{sec:relatedWork}

With the rise of software-defined networks, researchers have started exploring the benefits and challenges of
more adaptive network operations on many fronts~\cite{pieee19,SOSR2}. 
The consistent network update problem has already received much attention.
We refer to the extensive survey by Foerster et al.~\cite{FoersterSV19} for an overview.
There also exists interesting empirical studies on the topic of consistent network update, for example, Kuzniar~\etal~\cite{KuzniarPK15} 
showed the high variance in the timing of updates in switches. 
While approaches to maintaining logical properties such as loop-freedom and waypoint enforcement alone
are fairly well-understood in the literature, much less is known about algorithms that provably account for performance 
aspects such as congestion~\cite{FoersterSV19}.

In the seminal work in the area of consistent network update, Reitblatt~\etal~\cite{ReitblattFRSW12} introduced a fairly general two-phase approach 
to update networks while preserving reachability to the terminal. 
Mahajan and Wattenhofer~\cite{MahajanW13} have initiated the study of update mechanisms that do not require packet tagging. They introduced the first approach to maximize the number of links that can be updated in each round.
The problem was shown to be NP-hard
for a single terminal~\cite{ForsterW16,AmiriLMS16}, and the study has been extended for multiple terminals in~\cite{VanbeverVPFB12,ForsterW16}.

When there are multiple routes with different policies, it may further be important to minimize the number of interactions with an individual router~\cite{DudyczLS16}.
Recently, a few papers have focused on the synthesis of update schedules~\cite{Netsynth15,DidriksenJJKLLSS21, Snowcap21, Jiri2022, ifip22}, however,
in this paper, we are interested in the natural and well-studied objective of minimizing the number of rounds. Ludwig~\etal~\cite{LudwigMS15,FoersterLMS18} showed that finding a schedule that provides loop freedom in 3 rounds is NP-hard, and there is always a pair of flows that requires $\Theta(n)$ rounds. The authors of~\cite{LudwigDRS16} presented a mixed integer linear program to compute optimal solutions, however, without considering congestion; our formulation builds upon that and accounts for congestion aspects, as studied in~\cite{AmiriDSW18} for loop-free scenarios.   

On the other hand, relatively little is known about scheduling congestion-free updates, and most existing works revolve around heuristics~\cite{FoersterSV19,ZhengCSDWN17,SOSR1-Congest,HongKMZGNW13}. Amiri~\etal~\cite{AmiriDSW18} presented a first algorithmic result, devising a fixed-parameter tractable algorithm to update a fixed number of flows on directed acyclic graphs. The authors also showed that the problem is NP-hard in general, already for two flows. The authors later extended their results in~\cite{AmiriDP0W19}, considering more general but still acyclic graphs, and focusing on optimal solutions.
To the best of our knowledge, we are the first to observe and exploit the augmentation of links to speed up and improve the feasibility of update schedules. 

\section{Conclusion}\label{sec:conclusion}
This paper uncovered an interesting tradeoff between augmentation, speed and feasibility of network updates. 
We proved that $2$ times multiplicative augmentation is sufficient to make any update schedule feasible and provided insight into the complexity of scenarios with lower augmentation.
We further presented fast and optimal algorithms for finding consistent update schedules and empirically showed that 
the tradeoff between augmentation and speed is even better in practice.

In future works, it would be interesting to further explore this tradeoff considering additional consistency properties
(such as waypoint enforcement), or to explore extensions to scenarios supporting splittable flows.

\begin{acks}
This project has received funding from the
European Research Council (ERC)  under the European Union's Horizon 2020
research and innovation programme (Grant agreement No.\ 101019564
``The Design of Modern Fully Dynamic Data Structures (MoDynStruct)'')
and from the
Austrian Science Fund (FWF) project ``Fast Algorithms for a Reactive Network
Layer (ReactNet)'', P~33775-N, with additional funding from the \textit{netidee SCIENCE
Stiftung}, 2020--2024. This project is also supported by the German Federal Ministry of Education and Research (BMBF), 6G-RIC grant 16KISK020K, 2021-2025, as well as 
by the European Research Council (ERC) under grant agreement No. 864228 (AdjustNet), 2020-2025.

\hfill \includegraphics[width = 0.14\textwidth, ]{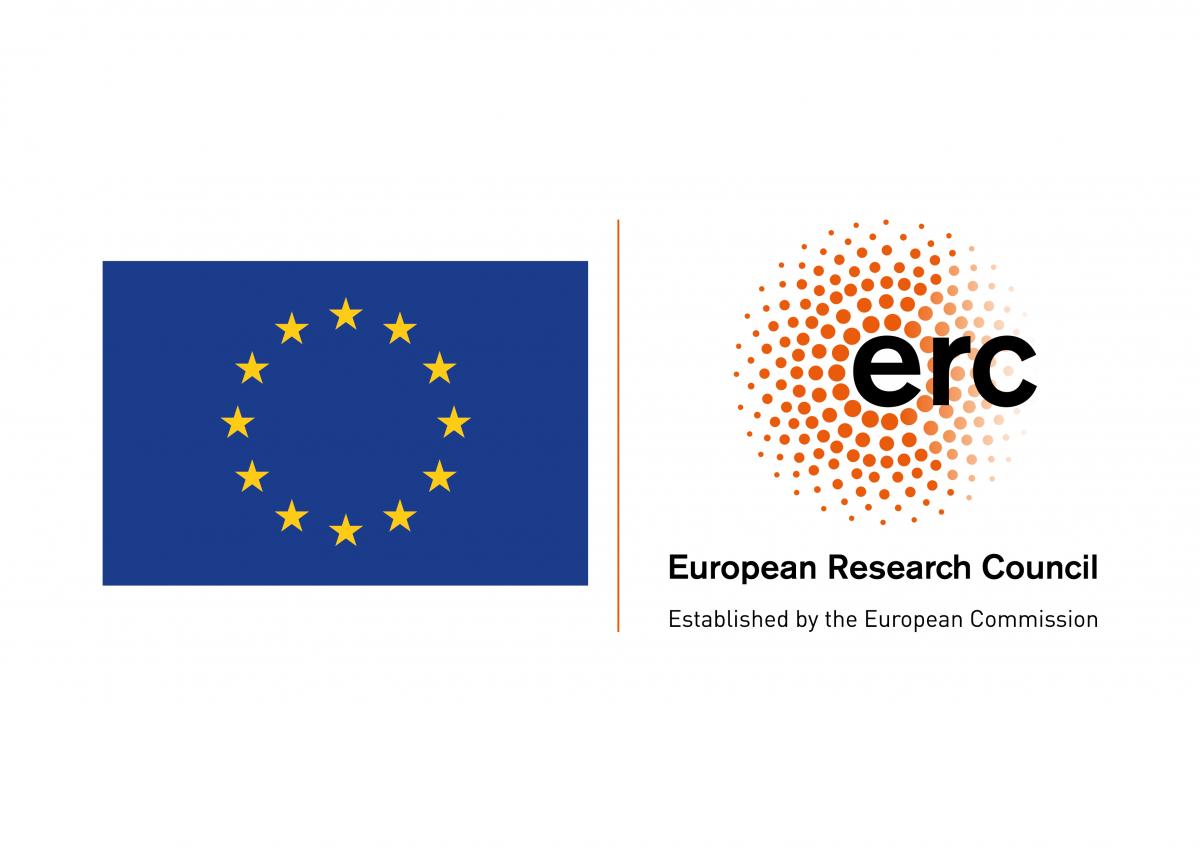}
\end{acks}

\bibliographystyle{ACM-Reference-Format}
\bibliography{rerouting}


\begin{thebibliography}{50}


\ifx \showCODEN    \undefined \def \showCODEN     #1{\unskip}     \fi
\ifx \showDOI      \undefined \def \showDOI       #1{#1}\fi
\ifx \showISBNx    \undefined \def \showISBNx     #1{\unskip}     \fi
\ifx \showISBNxiii \undefined \def \showISBNxiii  #1{\unskip}     \fi
\ifx \showISSN     \undefined \def \showISSN      #1{\unskip}     \fi
\ifx \showLCCN     \undefined \def \showLCCN      #1{\unskip}     \fi
\ifx \shownote     \undefined \def \shownote      #1{#1}          \fi
\ifx \showarticletitle \undefined \def \showarticletitle #1{#1}   \fi
\ifx \showURL      \undefined \def \showURL       {\relax}        \fi
\providecommand\bibfield[2]{#2}
\providecommand\bibinfo[2]{#2}
\providecommand\natexlab[1]{#1}
\providecommand\showeprint[2][]{arXiv:#2}

\bibitem[\protect\citeauthoryear{Amiri, Dudycz, Parham, Schmid, and
  Wiederrecht}{Amiri et~al\mbox{.}}{2019}]%
        {AmiriDP0W19}
\bibfield{author}{\bibinfo{person}{Saeed~Akhoondian Amiri},
  \bibinfo{person}{Szymon Dudycz}, \bibinfo{person}{Mahmoud Parham},
  \bibinfo{person}{Stefan Schmid}, {and} \bibinfo{person}{Sebastian
  Wiederrecht}.} \bibinfo{year}{2019}\natexlab{}.
\newblock \showarticletitle{On Polynomial-Time Congestion-Free Software-Defined
  Network Updates}. In \bibinfo{booktitle}{\emph{Proc. of the {IFIP} Networking
  Conference}}.
\newblock


\bibitem[\protect\citeauthoryear{Amiri, Dudycz, Schmid, and Wiederrecht}{Amiri
  et~al\mbox{.}}{2018}]%
        {AmiriDSW18}
\bibfield{author}{\bibinfo{person}{Saeed~Akhoondian Amiri},
  \bibinfo{person}{Szymon Dudycz}, \bibinfo{person}{Stefan Schmid}, {and}
  \bibinfo{person}{Sebastian Wiederrecht}.} \bibinfo{year}{2018}\natexlab{}.
\newblock \showarticletitle{Congestion-Free Rerouting of Flows on DAGs}. In
  \bibinfo{booktitle}{\emph{Proc. of the {ICALP}}}.
\newblock


\bibitem[\protect\citeauthoryear{Amiri, Ludwig, Marcinkowski, and Schmid}{Amiri
  et~al\mbox{.}}{2016}]%
        {AmiriLMS16}
\bibfield{author}{\bibinfo{person}{Saeed~Akhoondian Amiri},
  \bibinfo{person}{Arne Ludwig}, \bibinfo{person}{Jan Marcinkowski}, {and}
  \bibinfo{person}{Stefan Schmid}.} \bibinfo{year}{2016}\natexlab{}.
\newblock \showarticletitle{Transiently Consistent {SDN} Updates: Being Greedy
  is Hard}. In \bibinfo{booktitle}{\emph{Proc. of the {SIROCCO}}}.
\newblock


\bibitem[\protect\citeauthoryear{Avin, Ghobadi, Griner, and Schmid}{Avin
  et~al\mbox{.}}{2020}]%
        {AvinGG020}
\bibfield{author}{\bibinfo{person}{Chen Avin}, \bibinfo{person}{Manya Ghobadi},
  \bibinfo{person}{Chen Griner}, {and} \bibinfo{person}{Stefan Schmid}.}
  \bibinfo{year}{2020}\natexlab{}.
\newblock \showarticletitle{On the Complexity of Traffic Traces and
  Implications}.
\newblock \bibinfo{journal}{\emph{Proc. {ACM} Meas. Anal. Comput. Syst.}}
  (\bibinfo{year}{2020}).
\newblock


\bibitem[\protect\citeauthoryear{Beckett, Mahajan, Millstein, Padhye, and
  Walker}{Beckett et~al\mbox{.}}{2016}]%
        {BeckettMMPW16}
\bibfield{author}{\bibinfo{person}{Ryan Beckett}, \bibinfo{person}{Ratul
  Mahajan}, \bibinfo{person}{Todd~D. Millstein}, \bibinfo{person}{Jitendra
  Padhye}, {and} \bibinfo{person}{David Walker}.}
  \bibinfo{year}{2016}\natexlab{}.
\newblock \showarticletitle{Don't Mind the Gap: Bridging Network-wide
  Objectives and Device-level Configurations}. In
  \bibinfo{booktitle}{\emph{Proc. of the {ACM} {SIGCOMM}}}.
\newblock


\bibitem[\protect\citeauthoryear{Brandt, F{\"{o}}rster, and Wattenhofer}{Brandt
  et~al\mbox{.}}{2016}]%
        {BrandtFW16}
\bibfield{author}{\bibinfo{person}{Sebastian Brandt},
  \bibinfo{person}{Klaus{-}Tycho F{\"{o}}rster}, {and} \bibinfo{person}{Roger
  Wattenhofer}.} \bibinfo{year}{2016}\natexlab{}.
\newblock \showarticletitle{On consistent migration of flows in SDNs}. In
  \bibinfo{booktitle}{\emph{Proc. of the {IEEE} {INFOCOM}}}.
\newblock


\bibitem[\protect\citeauthoryear{Chirgwin}{Chirgwin}{2017}]%
        {chirgwin2017google}
\bibfield{author}{\bibinfo{person}{Richard Chirgwin}.}
  \bibinfo{year}{2017}\natexlab{}.
\newblock \showarticletitle{Google routing blunder sent japan’s internet dark
  on friday}.
\newblock \bibinfo{journal}{\emph{The Register}} (\bibinfo{year}{2017}).
\newblock


\bibitem[\protect\citeauthoryear{Christensen, Glavind, Schmid, and
  Srba}{Christensen et~al\mbox{.}}{2021}]%
        {christensen2021latte}
\bibfield{author}{\bibinfo{person}{Niels Christensen}, \bibinfo{person}{Mark
  Glavind}, \bibinfo{person}{Stefan Schmid}, {and}
  \bibinfo{person}{Ji{\v{r}}{\'\i} Srba}.} \bibinfo{year}{2021}\natexlab{}.
\newblock \showarticletitle{Latte: improving the latency of transiently
  consistent network update schedules}.
\newblock \bibinfo{journal}{\emph{{ACM} {SIGMETRICS} Performance Evaluation
  Review}} (\bibinfo{year}{2021}).
\newblock


\bibitem[\protect\citeauthoryear{Didriksen, Jensen, J{\o}nler, Katona, Lama,
  Lottrup, Shajarat, and Srba}{Didriksen et~al\mbox{.}}{2021}]%
        {DidriksenJJKLLSS21}
\bibfield{author}{\bibinfo{person}{M. Didriksen}, \bibinfo{person}{P.G.
  Jensen}, \bibinfo{person}{J.F. J{\o}nler}, \bibinfo{person}{A.-I. Katona},
  \bibinfo{person}{S.D.L. Lama}, \bibinfo{person}{F.B. Lottrup},
  \bibinfo{person}{S. Shajarat}, {and} \bibinfo{person}{J. Srba}.}
  \bibinfo{year}{2021}\natexlab{}.
\newblock \showarticletitle{Automatic Synthesis of Transiently Correct Network
  Updates via {Petri} Games}. In \bibinfo{booktitle}{\emph{Proc. of the {Petri
  Nets}}}.
\newblock


\bibitem[\protect\citeauthoryear{Dudycz, Ludwig, and Schmid}{Dudycz
  et~al\mbox{.}}{2016}]%
        {DudyczLS16}
\bibfield{author}{\bibinfo{person}{Szymon Dudycz}, \bibinfo{person}{Arne
  Ludwig}, {and} \bibinfo{person}{Stefan Schmid}.}
  \bibinfo{year}{2016}\natexlab{}.
\newblock \showarticletitle{Can't Touch This: Consistent Network Updates for
  Multiple Policies}. In \bibinfo{booktitle}{\emph{Proc. of {IEEE/IFIP}
  {DSN}}}.
\newblock


\bibitem[\protect\citeauthoryear{Feamster and Rexford}{Feamster and
  Rexford}{2018}]%
        {FeamsterR18}
\bibfield{author}{\bibinfo{person}{Nick Feamster} {and}
  \bibinfo{person}{Jennifer Rexford}.} \bibinfo{year}{2018}\natexlab{}.
\newblock \showarticletitle{Why (and How) Networks Should Run Themselves}. In
  \bibinfo{booktitle}{\emph{Proc. of the Applied Networking Research Workshop,
  {ANRW}}}.
\newblock


\bibitem[\protect\citeauthoryear{Filsfils, Nainar, Pignataro, Cardona, and
  Fran{\c{c}}ois}{Filsfils et~al\mbox{.}}{2015}]%
        {FilsfilsNPCF15}
\bibfield{author}{\bibinfo{person}{Clarence Filsfils},
  \bibinfo{person}{Nagendra~Kumar Nainar}, \bibinfo{person}{Carlos Pignataro},
  \bibinfo{person}{Juan~Camilo Cardona}, {and} \bibinfo{person}{Pierre
  Fran{\c{c}}ois}.} \bibinfo{year}{2015}\natexlab{}.
\newblock \showarticletitle{The Segment Routing Architecture}. In
  \bibinfo{booktitle}{\emph{Proc. of the {IEEE} {GLOBECOM}}}.
\newblock


\bibitem[\protect\citeauthoryear{Foerster, Ludwig, Marcinkowski, and
  Schmid}{Foerster et~al\mbox{.}}{2018}]%
        {FoersterLMS18}
\bibfield{author}{\bibinfo{person}{Klaus{-}Tycho Foerster},
  \bibinfo{person}{Arne Ludwig}, \bibinfo{person}{Jan Marcinkowski}, {and}
  \bibinfo{person}{Stefan Schmid}.} \bibinfo{year}{2018}\natexlab{}.
\newblock \showarticletitle{Loop-Free Route Updates for Software-Defined
  Networks}.
\newblock \bibinfo{journal}{\emph{{IEEE/ACM} Trans. Netw.}}
  (\bibinfo{year}{2018}).
\newblock


\bibitem[\protect\citeauthoryear{Foerster, Schmid, and Vissicchio}{Foerster
  et~al\mbox{.}}{2019}]%
        {FoersterSV19}
\bibfield{author}{\bibinfo{person}{Klaus{-}Tycho Foerster},
  \bibinfo{person}{Stefan Schmid}, {and} \bibinfo{person}{Stefano Vissicchio}.}
  \bibinfo{year}{2019}\natexlab{}.
\newblock \showarticletitle{Survey of Consistent Software-Defined Network
  Updates}.
\newblock \bibinfo{journal}{\emph{{IEEE} Commun. Surv. Tutorials}}
  (\bibinfo{year}{2019}).
\newblock


\bibitem[\protect\citeauthoryear{F{\"{o}}rster, Mahajan, and
  Wattenhofer}{F{\"{o}}rster et~al\mbox{.}}{2016}]%
        {ForsterMW16}
\bibfield{author}{\bibinfo{person}{Klaus{-}Tycho F{\"{o}}rster},
  \bibinfo{person}{Ratul Mahajan}, {and} \bibinfo{person}{Roger Wattenhofer}.}
  \bibinfo{year}{2016}\natexlab{}.
\newblock \showarticletitle{Consistent updates in software defined networks: On
  dependencies, loop freedom, and blackholes}. In
  \bibinfo{booktitle}{\emph{Proc. of the {IFIP} Networking Conference}}.
\newblock


\bibitem[\protect\citeauthoryear{F{\"{o}}rster and Wattenhofer}{F{\"{o}}rster
  and Wattenhofer}{2016}]%
        {ForsterW16}
\bibfield{author}{\bibinfo{person}{Klaus{-}Tycho F{\"{o}}rster} {and}
  \bibinfo{person}{Roger Wattenhofer}.} \bibinfo{year}{2016}\natexlab{}.
\newblock \showarticletitle{The Power of Two in Consistent Network Updates:
  Hard Loop Freedom, Easy Flow Migration}. In \bibinfo{booktitle}{\emph{Proc.
  of the {ICCCN}}}.
\newblock


\bibitem[\protect\citeauthoryear{{Gurobi Optimization, LLC}}{{Gurobi
  Optimization, LLC}}{2021}]%
        {gurobi}
\bibfield{author}{\bibinfo{person}{{Gurobi Optimization, LLC}}.}
  \bibinfo{year}{2021}\natexlab{}.
\newblock \bibinfo{title}{{Gurobi Optimizer Reference Manual}}.
\newblock
\newblock
\urldef\tempurl%
\url{"https://www.gurobi.com"}
\showURL{%
\tempurl}


\bibitem[\protect\citeauthoryear{Hagberg, Swart, and S~Chult}{Hagberg
  et~al\mbox{.}}{2008}]%
        {networkx}
\bibfield{author}{\bibinfo{person}{Aric Hagberg}, \bibinfo{person}{Pieter
  Swart}, {and} \bibinfo{person}{Daniel S~Chult}.}
  \bibinfo{year}{2008}\natexlab{}.
\newblock \bibinfo{booktitle}{\emph{Exploring network structure, dynamics, and
  function using NetworkX}}.
\newblock \bibinfo{type}{{T}echnical {R}eport}.
\newblock


\bibitem[\protect\citeauthoryear{Harris, Millman, van~der Walt, Gommers,
  Virtanen, Cournapeau, Wieser, Taylor, Berg, Smith, Kern, Picus, Hoyer, van
  Kerkwijk, Brett, Haldane, del R{\'{i}}o, Wiebe, Peterson,
  G{\'{e}}rard-Marchant, Sheppard, Reddy, Weckesser, Abbasi, Gohlke, and
  Oliphant}{Harris et~al\mbox{.}}{2020}]%
        {numpy}
\bibfield{author}{\bibinfo{person}{Charles~R. Harris},
  \bibinfo{person}{K.~Jarrod Millman}, \bibinfo{person}{St{\'{e}}fan~J. van~der
  Walt}, \bibinfo{person}{Ralf Gommers}, \bibinfo{person}{Pauli Virtanen},
  \bibinfo{person}{David Cournapeau}, \bibinfo{person}{Eric Wieser},
  \bibinfo{person}{Julian Taylor}, \bibinfo{person}{Sebastian Berg},
  \bibinfo{person}{Nathaniel~J. Smith}, \bibinfo{person}{Robert Kern},
  \bibinfo{person}{Matti Picus}, \bibinfo{person}{Stephan Hoyer},
  \bibinfo{person}{Marten~H. van Kerkwijk}, \bibinfo{person}{Matthew Brett},
  \bibinfo{person}{Allan Haldane}, \bibinfo{person}{Jaime~Fern{\'{a}}ndez del
  R{\'{i}}o}, \bibinfo{person}{Mark Wiebe}, \bibinfo{person}{Pearu Peterson},
  \bibinfo{person}{Pierre G{\'{e}}rard-Marchant}, \bibinfo{person}{Kevin
  Sheppard}, \bibinfo{person}{Tyler Reddy}, \bibinfo{person}{Warren Weckesser},
  \bibinfo{person}{Hameer Abbasi}, \bibinfo{person}{Christoph Gohlke}, {and}
  \bibinfo{person}{Travis~E. Oliphant}.} \bibinfo{year}{2020}\natexlab{}.
\newblock \showarticletitle{Array programming with {NumPy}}.
\newblock \bibinfo{journal}{\emph{Nature}} (\bibinfo{year}{2020}).
\newblock


\bibitem[\protect\citeauthoryear{He, Khalid, Gember{-}Jacobson, Das, Prakash,
  Akella, Li, and Thottan}{He et~al\mbox{.}}{2015}]%
        {UpToSeconds}
\bibfield{author}{\bibinfo{person}{Keqiang He}, \bibinfo{person}{Junaid
  Khalid}, \bibinfo{person}{Aaron Gember{-}Jacobson}, \bibinfo{person}{Sourav
  Das}, \bibinfo{person}{Chaithan Prakash}, \bibinfo{person}{Aditya Akella},
  \bibinfo{person}{Li~Erran Li}, {and} \bibinfo{person}{Marina Thottan}.}
  \bibinfo{year}{2015}\natexlab{}.
\newblock \showarticletitle{Measuring control plane latency in SDN-enabled
  switches}. In \bibinfo{booktitle}{\emph{Proc. of 1st {ACM} {SIGCOMM}
  Symposium on Software Defined Networking Research, {SOSR}}}.
\newblock


\bibitem[\protect\citeauthoryear{Hong, Kandula, Mahajan, Zhang, Gill, Nanduri,
  and Wattenhofer}{Hong et~al\mbox{.}}{2013}]%
        {HongKMZGNW13}
\bibfield{author}{\bibinfo{person}{Chi{-}Yao Hong}, \bibinfo{person}{Srikanth
  Kandula}, \bibinfo{person}{Ratul Mahajan}, \bibinfo{person}{Ming Zhang},
  \bibinfo{person}{Vijay Gill}, \bibinfo{person}{Mohan Nanduri}, {and}
  \bibinfo{person}{Roger Wattenhofer}.} \bibinfo{year}{2013}\natexlab{}.
\newblock \showarticletitle{Achieving high utilization with software-driven
  {WAN}}. In \bibinfo{booktitle}{\emph{Proc. of the {ACM} {SIGCOMM}
  Conference}}.
\newblock


\bibitem[\protect\citeauthoryear{Hunter}{Hunter}{2007}]%
        {Matplotlib}
\bibfield{author}{\bibinfo{person}{John~D. Hunter}.}
  \bibinfo{year}{2007}\natexlab{}.
\newblock \showarticletitle{Matplotlib: {A} 2D Graphics Environment}.
\newblock \bibinfo{journal}{\emph{Comput. Sci. Eng.}} (\bibinfo{year}{2007}).
\newblock


\bibitem[\protect\citeauthoryear{Jacobson}{Jacobson}{2008}]%
        {Jacobson88}
\bibfield{author}{\bibinfo{person}{Van Jacobson}.}
  \bibinfo{year}{2008}\natexlab{}.
\newblock \showarticletitle{Congestion avoidance and control}. In
  \bibinfo{booktitle}{\emph{Proc. of the {ACM} {SIGCOMM}}}.
\newblock


\bibitem[\protect\citeauthoryear{Jin, Liu, Gandhi, Kandula, Mahajan, Zhang,
  Rexford, and Wattenhofer}{Jin et~al\mbox{.}}{2014}]%
        {JinLGKMZRW14}
\bibfield{author}{\bibinfo{person}{Xin Jin}, \bibinfo{person}{Hongqiang~Harry
  Liu}, \bibinfo{person}{Rohan Gandhi}, \bibinfo{person}{Srikanth Kandula},
  \bibinfo{person}{Ratul Mahajan}, \bibinfo{person}{Ming Zhang},
  \bibinfo{person}{Jennifer Rexford}, {and} \bibinfo{person}{Roger
  Wattenhofer}.} \bibinfo{year}{2014}\natexlab{}.
\newblock \showarticletitle{Dynamic scheduling of network updates}. In
  \bibinfo{booktitle}{\emph{Proc. of the {ACM} {SIGCOMM}}}.
\newblock


\bibitem[\protect\citeauthoryear{Johansen, Kaer, Madsen, Nielsen, Srba, and
  Tollund}{Johansen et~al\mbox{.}}{2022}]%
        {Jiri2022}
\bibfield{author}{\bibinfo{person}{N.S. Johansen}, \bibinfo{person}{L.B. Kaer},
  \bibinfo{person}{A.L. Madsen}, \bibinfo{person}{K.O. Nielsen},
  \bibinfo{person}{J. Srba}, {and} \bibinfo{person}{R.G. Tollund}.}
  \bibinfo{year}{2022}\natexlab{}.
\newblock \showarticletitle{Kaki: Concurrent Update Synthesis for Regular
  Policies via Petri Games}. In \bibinfo{booktitle}{\emph{Proc. of the
  International Conference on Integrated Formal Methods, {iFM'22}}}.
\newblock


\bibitem[\protect\citeauthoryear{Karp}{Karp}{1972}]%
        {3SAT}
\bibfield{author}{\bibinfo{person}{Richard~M. Karp}.}
  \bibinfo{year}{1972}\natexlab{}.
\newblock \showarticletitle{Reducibility Among Combinatorial Problems}. In
  \bibinfo{booktitle}{\emph{Proc. of a symposium on the Complexity of Computer
  Computations}} \emph{(\bibinfo{series}{The {IBM} Research Symposia Series})}.
\newblock


\bibitem[\protect\citeauthoryear{Kellerer, Kalmbach, Blenk, Basta, Reisslein,
  and Schmid}{Kellerer et~al\mbox{.}}{2019}]%
        {pieee19}
\bibfield{author}{\bibinfo{person}{Wolfgang Kellerer}, \bibinfo{person}{Patrick
  Kalmbach}, \bibinfo{person}{Andreas Blenk}, \bibinfo{person}{Arsany Basta},
  \bibinfo{person}{Martin Reisslein}, {and} \bibinfo{person}{Stefan Schmid}.}
  \bibinfo{year}{2019}\natexlab{}.
\newblock \showarticletitle{Adaptable and Data-Driven Softwarized Networks:
  Review, Opportunities, and Challenges}. In \bibinfo{booktitle}{\emph{Proc. of
  the {IEEE}, {PIEEE}}}.
\newblock


\bibitem[\protect\citeauthoryear{Knight, Nguyen, Falkner, Bowden, and
  Roughan}{Knight et~al\mbox{.}}{2011}]%
        {TopologyZoo}
\bibfield{author}{\bibinfo{person}{Simon Knight}, \bibinfo{person}{Hung~X.
  Nguyen}, \bibinfo{person}{Nick Falkner}, \bibinfo{person}{Rhys~Alistair
  Bowden}, {and} \bibinfo{person}{Matthew Roughan}.}
  \bibinfo{year}{2011}\natexlab{}.
\newblock \showarticletitle{The Internet Topology Zoo}.
\newblock \bibinfo{journal}{\emph{{IEEE} J. Sel. Areas Commun.}}
  (\bibinfo{year}{2011}).
\newblock


\bibitem[\protect\citeauthoryear{Kuzniar, Peres{\'{\i}}ni, and Kostic}{Kuzniar
  et~al\mbox{.}}{2015}]%
        {KuzniarPK15}
\bibfield{author}{\bibinfo{person}{Maciej Kuzniar}, \bibinfo{person}{Peter
  Peres{\'{\i}}ni}, {and} \bibinfo{person}{Dejan Kostic}.}
  \bibinfo{year}{2015}\natexlab{}.
\newblock \showarticletitle{What You Need to Know About {SDN} Flow Tables}. In
  \bibinfo{booktitle}{\emph{Proc. of the {PAM}}}.
\newblock


\bibitem[\protect\citeauthoryear{Laoutaris, Smaragdakis, Rodriguez, and
  Sundaram}{Laoutaris et~al\mbox{.}}{2009}]%
        {laoutaris2009delay}
\bibfield{author}{\bibinfo{person}{Nikolaos Laoutaris},
  \bibinfo{person}{Georgios Smaragdakis}, \bibinfo{person}{Pablo Rodriguez},
  {and} \bibinfo{person}{Ravi Sundaram}.} \bibinfo{year}{2009}\natexlab{}.
\newblock \showarticletitle{Delay tolerant bulk data transfers on the
  internet}. In \bibinfo{booktitle}{\emph{Proc. of the {ACM}
  {SIGMETRICS/Performance}}}.
\newblock


\bibitem[\protect\citeauthoryear{Lev, Pippenger, and Valiant}{Lev
  et~al\mbox{.}}{1981}]%
        {LevPV81}
\bibfield{author}{\bibinfo{person}{Gavriela~Freund Lev},
  \bibinfo{person}{Nicholas Pippenger}, {and} \bibinfo{person}{Leslie~G.
  Valiant}.} \bibinfo{year}{1981}\natexlab{}.
\newblock \showarticletitle{A Fast Parallel Algorithm for Routing in
  Permutation Networks}.
\newblock \bibinfo{journal}{\emph{Proc. of the {IEEE} Trans. Computers}}
  (\bibinfo{year}{1981}).
\newblock


\bibitem[\protect\citeauthoryear{Liaskos, Mamatas, Pourdamghani, Tsioliaridou,
  Ioannidis, Pitsillides, Schmid, and Akyildi}{Liaskos et~al\mbox{.}}{2022}]%
        {pieee22}
\bibfield{author}{\bibinfo{person}{Christos Liaskos}, \bibinfo{person}{Lefteris
  Mamatas}, \bibinfo{person}{Arash Pourdamghani}, \bibinfo{person}{Atsioli
  Tsioliaridou}, \bibinfo{person}{Sotiris Ioannidis}, \bibinfo{person}{Andreas
  Pitsillides}, \bibinfo{person}{Stefan Schmid}, {and} \bibinfo{person}{Ian~F.
  Akyildi}.} \bibinfo{year}{2022}\natexlab{}.
\newblock \showarticletitle{Software-Defined Reconfigurable Intelligent
  Surfaces: From Theory to End-to-End Implementation}. In
  \bibinfo{booktitle}{\emph{Proc. of the IEEE (PIEEE)}}.
\newblock


\bibitem[\protect\citeauthoryear{Liu, Wu, Zhang, Yuan, Wattenhofer, and
  Maltz}{Liu et~al\mbox{.}}{2013}]%
        {liu2013zupdate}
\bibfield{author}{\bibinfo{person}{Hongqiang~Harry Liu}, \bibinfo{person}{Xin
  Wu}, \bibinfo{person}{Ming Zhang}, \bibinfo{person}{Lihua Yuan},
  \bibinfo{person}{Roger Wattenhofer}, {and} \bibinfo{person}{David~A. Maltz}.}
  \bibinfo{year}{2013}\natexlab{}.
\newblock \showarticletitle{zUpdate: updating data center networks with zero
  loss}. In \bibinfo{booktitle}{\emph{Proc. of the {ACM} {SIGCOMM}}}.
\newblock


\bibitem[\protect\citeauthoryear{Ludwig, Dudycz, Rost, and Schmid}{Ludwig
  et~al\mbox{.}}{2016}]%
        {LudwigDRS16}
\bibfield{author}{\bibinfo{person}{Arne Ludwig}, \bibinfo{person}{Szymon
  Dudycz}, \bibinfo{person}{Matthias Rost}, {and} \bibinfo{person}{Stefan
  Schmid}.} \bibinfo{year}{2016}\natexlab{}.
\newblock \showarticletitle{Transiently Secure Network Updates}. In
  \bibinfo{booktitle}{\emph{Proc. of {ACM} {SIGMETRICS}}}.
\newblock


\bibitem[\protect\citeauthoryear{Ludwig, Marcinkowski, and Schmid}{Ludwig
  et~al\mbox{.}}{2015}]%
        {LudwigMS15}
\bibfield{author}{\bibinfo{person}{Arne Ludwig}, \bibinfo{person}{Jan
  Marcinkowski}, {and} \bibinfo{person}{Stefan Schmid}.}
  \bibinfo{year}{2015}\natexlab{}.
\newblock \showarticletitle{Scheduling Loop-free Network Updates: It's Good to
  Relax!}. In \bibinfo{booktitle}{\emph{Proc. of the {ACM} {PODC}}}.
\newblock


\bibitem[\protect\citeauthoryear{Ludwig, Rost, Foucard, and Schmid}{Ludwig
  et~al\mbox{.}}{2014}]%
        {ludwig2014good}
\bibfield{author}{\bibinfo{person}{Arne Ludwig}, \bibinfo{person}{Matthias
  Rost}, \bibinfo{person}{Damien Foucard}, {and} \bibinfo{person}{Stefan
  Schmid}.} \bibinfo{year}{2014}\natexlab{}.
\newblock \showarticletitle{Good network updates for bad packets: Waypoint
  enforcement beyond destination-based routing policies}. In
  \bibinfo{booktitle}{\emph{Proc. of the {ACM} {HotNets}}}.
\newblock


\bibitem[\protect\citeauthoryear{Mahajan and Wattenhofer}{Mahajan and
  Wattenhofer}{2013}]%
        {MahajanW13}
\bibfield{author}{\bibinfo{person}{Ratul Mahajan} {and} \bibinfo{person}{Roger
  Wattenhofer}.} \bibinfo{year}{2013}\natexlab{}.
\newblock \showarticletitle{On consistent updates in software defined
  networks}. In \bibinfo{booktitle}{\emph{Proc. of {ACM} {HotNets}}}.
\newblock


\bibitem[\protect\citeauthoryear{McClurg, Hojjat, {\v{C}}ern{\`y}, and
  Foster}{McClurg et~al\mbox{.}}{2015a}]%
        {mcclurg2015efficient}
\bibfield{author}{\bibinfo{person}{Jedidiah McClurg}, \bibinfo{person}{Hossein
  Hojjat}, \bibinfo{person}{Pavol {\v{C}}ern{\`y}}, {and} \bibinfo{person}{Nate
  Foster}.} \bibinfo{year}{2015}\natexlab{a}.
\newblock \showarticletitle{Efficient synthesis of network updates}.
\newblock \bibinfo{journal}{\emph{{ACM} Sigplan Notices}}
  (\bibinfo{year}{2015}).
\newblock


\bibitem[\protect\citeauthoryear{McClurg, Hojjat, Cern{\'{y}}, and
  Foster}{McClurg et~al\mbox{.}}{2015b}]%
        {Netsynth15}
\bibfield{author}{\bibinfo{person}{Jedidiah McClurg}, \bibinfo{person}{Hossein
  Hojjat}, \bibinfo{person}{Pavol Cern{\'{y}}}, {and} \bibinfo{person}{Nate
  Foster}.} \bibinfo{year}{2015}\natexlab{b}.
\newblock \showarticletitle{Efficient synthesis of network updates}. In
  \bibinfo{booktitle}{\emph{Proceedings of the {ACM} {SIGPLAN}}}.
\newblock


\bibitem[\protect\citeauthoryear{Miller, Tucker, and Zemlin}{Miller
  et~al\mbox{.}}{1960}]%
        {miller1960integer}
\bibfield{author}{\bibinfo{person}{Clair~E Miller}, \bibinfo{person}{Albert~W
  Tucker}, {and} \bibinfo{person}{Richard~A Zemlin}.}
  \bibinfo{year}{1960}\natexlab{}.
\newblock \showarticletitle{Integer programming formulation of traveling
  salesman problems}.
\newblock \bibinfo{journal}{\emph{J. of the ACM, {JACM}}}
  (\bibinfo{year}{1960}).
\newblock


\bibitem[\protect\citeauthoryear{Nguyen, Chiesa, and Canini}{Nguyen
  et~al\mbox{.}}{2017}]%
        {SOSR1-Congest}
\bibfield{author}{\bibinfo{person}{Thanh~Dang Nguyen}, \bibinfo{person}{Marco
  Chiesa}, {and} \bibinfo{person}{Marco Canini}.}
  \bibinfo{year}{2017}\natexlab{}.
\newblock \showarticletitle{Decentralized Consistent Updates in {SDN}}. In
  \bibinfo{booktitle}{\emph{Proc. of {ACM} {SIGCOMM} the Symposium on {SDN}
  Research, {SOSR}}}.
\newblock


\bibitem[\protect\citeauthoryear{Reitblatt, Foster, Rexford, Schlesinger, and
  Walker}{Reitblatt et~al\mbox{.}}{2012}]%
        {ReitblattFRSW12}
\bibfield{author}{\bibinfo{person}{Mark Reitblatt}, \bibinfo{person}{Nate
  Foster}, \bibinfo{person}{Jennifer Rexford}, \bibinfo{person}{Cole
  Schlesinger}, {and} \bibinfo{person}{David Walker}.}
  \bibinfo{year}{2012}\natexlab{}.
\newblock \showarticletitle{Abstractions for network update}. In
  \bibinfo{booktitle}{\emph{Proc. of the {ACM} {SIGCOMM}}}.
\newblock


\bibitem[\protect\citeauthoryear{\relax Duluth News~Tribune}{\relax Duluth
  News~Tribune}{2018}]%
        {ems1_2018}
\bibfield{author}{\bibinfo{person}{\relax Duluth News~Tribune}.}
  \bibinfo{year}{2018}\natexlab{}.
\newblock \bibinfo{title}{Officials: Human error to blame in Minn. 911 outage}.
\newblock
\newblock
\urldef\tempurl%
\url{https://www.ems1.com/911/articles/officials-human-error-to-blame-in-minn-911-outage-xjwcEfhzmbDD8tub/}
\showURL{%
\tempurl}


\bibitem[\protect\citeauthoryear{Roy, Zeng, Bagga, Porter, and Snoeren}{Roy
  et~al\mbox{.}}{2015}]%
        {RoyZBPS15}
\bibfield{author}{\bibinfo{person}{Arjun Roy}, \bibinfo{person}{Hongyi Zeng},
  \bibinfo{person}{Jasmeet Bagga}, \bibinfo{person}{George Porter}, {and}
  \bibinfo{person}{Alex~C. Snoeren}.} \bibinfo{year}{2015}\natexlab{}.
\newblock \showarticletitle{Inside the Social Network's (Datacenter) Network}.
  In \bibinfo{booktitle}{\emph{Proc. of the {ACM} {SIGCOMM}}}.
\newblock


\bibitem[\protect\citeauthoryear{Saur, Collard, Foster, Guha, Vanbever, and
  Hicks}{Saur et~al\mbox{.}}{2016}]%
        {SOSR2}
\bibfield{author}{\bibinfo{person}{Karla Saur}, \bibinfo{person}{Joseph~M.
  Collard}, \bibinfo{person}{Nate Foster}, \bibinfo{person}{Arjun Guha},
  \bibinfo{person}{Laurent Vanbever}, {and} \bibinfo{person}{Michael~W.
  Hicks}.} \bibinfo{year}{2016}\natexlab{}.
\newblock \showarticletitle{Safe and Flexible Controller Upgrades for SDNs}. In
  \bibinfo{booktitle}{\emph{Proc. of {ACM} {SIGCOMM} the Symposium on {SDN}
  Research, {SOSR}}}.
\newblock


\bibitem[\protect\citeauthoryear{Schmid, Schrenk, and Torralba}{Schmid
  et~al\mbox{.}}{2022}]%
        {ifip22}
\bibfield{author}{\bibinfo{person}{Stefan Schmid}, \bibinfo{person}{Bernhard
  Schrenk}, {and} \bibinfo{person}{Alvaro Torralba}.}
  \bibinfo{year}{2022}\natexlab{}.
\newblock \showarticletitle{NetStack: A Game Approach to Synthesizing
  Consistent Network Updates}. In \bibinfo{booktitle}{\emph{Proc. IFIP
  Networking}}.
\newblock


\bibitem[\protect\citeauthoryear{Schneider, Birkner, and Vanbever}{Schneider
  et~al\mbox{.}}{2021}]%
        {Snowcap21}
\bibfield{author}{\bibinfo{person}{Tibor Schneider},
  \bibinfo{person}{R{\"{u}}diger Birkner}, {and} \bibinfo{person}{Laurent
  Vanbever}.} \bibinfo{year}{2021}\natexlab{}.
\newblock \showarticletitle{Snowcap: synthesizing network-wide configuration
  updates}. In \bibinfo{booktitle}{\emph{Proc. of the {ACM} {SIGCOMM}}}.
\newblock


\bibitem[\protect\citeauthoryear{Stevens}{Stevens}{1997}]%
        {SlowStart}
\bibfield{author}{\bibinfo{person}{W.~Richard Stevens}.}
  \bibinfo{year}{1997}\natexlab{}.
\newblock \showarticletitle{{TCP} Slow Start, Congestion Avoidance, Fast
  Retransmit, and Fast Recovery Algorithms}.
\newblock \bibinfo{journal}{\emph{{RFC}}} (\bibinfo{year}{1997}).
\newblock


\bibitem[\protect\citeauthoryear{Vanbever, Vissicchio, Pelsser, Fran{\c{c}}ois,
  and Bonaventure}{Vanbever et~al\mbox{.}}{2012}]%
        {VanbeverVPFB12}
\bibfield{author}{\bibinfo{person}{Laurent Vanbever}, \bibinfo{person}{Stefano
  Vissicchio}, \bibinfo{person}{Cristel Pelsser}, \bibinfo{person}{Pierre
  Fran{\c{c}}ois}, {and} \bibinfo{person}{Olivier Bonaventure}.}
  \bibinfo{year}{2012}\natexlab{}.
\newblock \showarticletitle{Lossless migrations of link-state IGPs}.
\newblock \bibinfo{journal}{\emph{{IEEE/ACM} Trans. Netw.}}
  (\bibinfo{year}{2012}).
\newblock


\bibitem[\protect\citeauthoryear{Zheng, Chen, Schmid, Dai, Wu, and Ni}{Zheng
  et~al\mbox{.}}{2017}]%
        {ZhengCSDWN17}
\bibfield{author}{\bibinfo{person}{Jiaqi Zheng}, \bibinfo{person}{Guihai Chen},
  \bibinfo{person}{Stefan Schmid}, \bibinfo{person}{Haipeng Dai},
  \bibinfo{person}{Jie Wu}, {and} \bibinfo{person}{Qiang Ni}.}
  \bibinfo{year}{2017}\natexlab{}.
\newblock \showarticletitle{Scheduling Congestion- and Loop-Free Network Update
  in Timed SDNs}.
\newblock \bibinfo{journal}{\emph{{IEEE} J. Sel. Areas Commun.}}
  (\bibinfo{year}{2017}).
\newblock


\end{thebibliography}

\appendix
\section{Omitted proofs}
\label{appendix: omitted}
In this section, we detail the proofs of Theorems~\ref{thm:2-eps is nph} and~\ref{thm:additive is nph} that prove the NP-hardness of finding valid schedules with a limited augmentation.
\multi
Given a 3-CNF formula and a constant $0<\epsilon\leq 1/3$, we build a graph and define flow pairs on it.
The core of the construction is variable gadgets, 
one for each variable $x_j$ in the formula, 
with two flow pairs representing truth values, such that one of them must be updated before the other.
A choice to update the true flow first or the false flow first, for each variable, implies a corresponding assignment of the value true or false to the variable, giving a satisfying assignment for the formula.
In addition, the construction contains another gadget, that guarantees that the choice of flows to update first indeed satisfies each of the clauses.

\paragraph*{The reduction}
The reduction graph is composed of the nodes $s$ and $t$, a pair of nodes $u^i,v^i$ for each clause $C^i$, a pair of nodes $u^i_j,v^i_j$ for each occurrence of a variable $x_j$ (negated or not) in a clause $C^i$,
and, for each variable $x_j$, the nodes $w^1_j,w^2_j$ and $w_j(0),\ldots,w_j(\sqrt a)$, where 
$
a=\left(\left\lfloor\frac{1}{2\epsilon}\right\rfloor+1\right)^2
$.
The graph edges and their capacities are defined as part of the flow definitions, next.

\begin{figure}
	\centering
	\includegraphics[scale=.8]{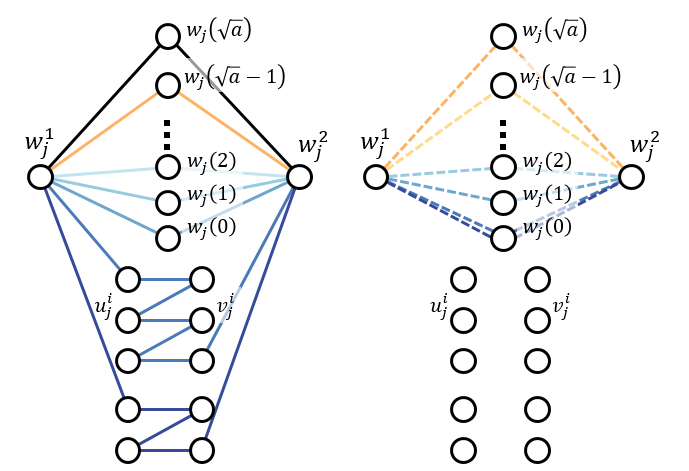}
	\caption{The variable gadget for variable $x_j$ which appears in two clauses without negation and in three clauses negated.
	Left: original flows; right: updated flows; edges without flows are omitted.
	The original flows, from bottom up, have sizes $a,a,2a, 2a+\sqrt a,2a+2\sqrt a,\ldots, 2a+(\sqrt a-1)\sqrt a, 3a$.}
	\label{fig: nph of 2-eps bit gadget}
\end{figure}

\paragraph{The variable gadget}
For each variable $x_j$, we define a gadget --- see Figure~\ref{fig: nph of 2-eps bit gadget}.
The gadget is composed of paths from $w^1_j$ to $w^2_j$, of different capacities and flows.

First, a path 
$(w^1_j,u^{i_1}_j,v^{i_1}_j\ldots ,u^{i_{\max(j)}}_j,v^{i_{\max(j)}}_j,	w^2_j)$ 
from $w^1_j$, through each edge $(u^{i}_j,v^{i}_j)$ such that $x_j$ appears in $C^i$ without negation (in an arbitrary order), and to $w^2_j$;
and, a similar path through each edge $(u^{i}_j,v^{i}_j)$ such that $x_j$ appears in $C^i$ negated. 
All edges in both paths have capacity $a$, and initial flows $F_{true}$ and $F_{false}$ respectively, of demand $a$ each.
Then, for $k=0,\ldots,\sqrt a$, the gadget contains a path $(w^1_j,w_j(k),w^2_j)$ with edges of capacity $2a+k\sqrt a$, and flow $F(k)$ of demand $2a+k\sqrt a$ as well;
the flow $F(\sqrt{a})$ will also be referred to as the \emph{blocking flow}.
The updated flows for both $F_{true}$ and $F_{false}$ are flows through $(w^1_j,w_j(0),w^2_j)$.
For $k=0,\ldots,\sqrt{a}-1$, the updated flow for $F(k)$ is on the path
$(w^1_j,w_j(k+1),w^2_j)$, i.e., the old path of the flow $F(k+1)$.
The updated flow for $F(\sqrt{a})$ is through a different gadget, described next.

\begin{figure}
	\centering
	\includegraphics[scale=1]{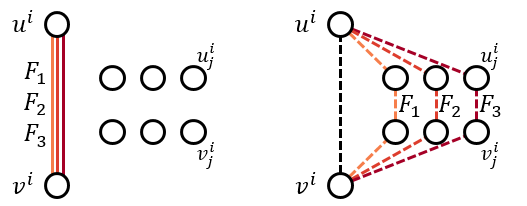}
	\caption{The clause gadget for a clause $C^i$ which contains the variable $x_j$.
	Left: original flows; right: updated flows; edges without flows are omitted.
	The three flows $F_1,F_2,F_3$ appear in both figures and have demand $a$ each;
	the updated blocking flow, which appears only on the right, has demand $3a$, and its original flow appears in the variable gadget (Figure~\ref{fig: nph of 2-eps blocking flow}).}
	\label{fig: nph of 2-eps clause gadget}
\end{figure}

\begin{figure}
	\centering
	\includegraphics[scale=.8,trim=0 0 0 10, clip]{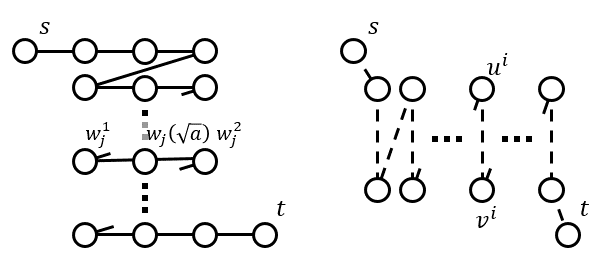}
	\caption{The blocking flow.
	Left: the old flow, going through all the variable gadgets;
	right: the updated flow, going through all the clause gadgets. The flow has demand $3a$, and this is also the capacity of all the edges in the figure.}
	\label{fig: nph of 2-eps blocking flow}
\end{figure}

\paragraph{The clause gadget}
For a clause $C^i$, we build a gadget composed of three flows $F_1,F_2,F_3$, which are initially identical ---  each flow has demand $a$, and they all go through the edge $(u^i,v^i)$, which has capacity $3a$.
The updated flows for these flows are through the paths 
$(u^i,u^i_j,v^i_j,v^i)$ for the three variables $x_j$ appearing in the clause $C^i$ (negated or not), one flow on each path; the edges of this path has capacity~$a$.
Clauses that contain both a variable and its negation are omitted from the construction, as they are always satisfied --- see Figure~\ref{fig: nph of 2-eps clause gadget}.

In addition, the edge $(u^i,v^i)$ is a part of the updated flow of the blocking flow, in a way described next.

\paragraph{Putting everything together}
We now explain how the flows go through the different gadgets.
Each of the flows defined above, i.e., $F_{true}, F_{false}, F_1,F_2,F_3$ and the flows $F(k)$ for $k=0,\ldots,\sqrt{a}$ all start from $s$, and then traverse all the relevant gadgets in an arbitrary order. 
For example, the flow $F_{true}$ goes from $s$ to $w^1_j$ for some $j$, through the gadget of $x_j$ to $w^2_j$, then to $w^1_{j'}$ for some $j'$ and so on, and finally from some $w^2_{j''}$ to $t$.
The corresponding updated flow goes through the same gadgets in the same order, and inside each gadget $j$ goes through $(w^1_j,w_j(0),w^2_j)$, as described above. 
The flows $F_{false}$ and $F(k)$ for $k=0,\ldots,\sqrt{a}$ are all connected in a similar manner.
Each of the flows $F_1,F_2,F_3$ go from $s$ to $u^i$ for some clause $C^i$, to $v^i$ through the gadget of $C^i$, to $v^{i'}$, and so on, until it leaves the last $v^{i''}$ to $t$.
The updated flows are similarly connected.
We assume there are edges connecting $s$ and $t$ to the gadgets and between the gadgets, where each edge has enough capacity to contain all the initial and updated flows assigned to it; note that the flows on these edges (except for $F(\sqrt a)$) are not changed by the updates.
The blocking flow $F(\sqrt a)$
initially traverses all the variable gadgets, and when updated, traverses all the clause gadgets, as depicted in Figure~\ref{fig: nph of 2-eps blocking flow}.
Except for this blocking flow, each flow can be seen as is split into segments, one for each gadget, and the updates on one gadget are independent the updates in others.

\paragraph*{Proof of the reduction}
Using the above construction, we show that deciding if a $(\times(2-\epsilon))$-valid update schedule exists is NP-hard. 
\begin{proof}[Proof of Theorem~\ref{thm:2-eps is nph}]
	Consider a 3-CNF formula, a constant $0<\epsilon < 1/3$, and define $a$, the graph and the flow pairs as described above.
	We show that the formula has a satisfying assignment if and only if the graph has a $(\times(2-\epsilon))$-valid update schedule.
	
	Assume the formula has a satisfying assignment, and define an update schedule in the following order:
	in each variable clause, update $F_{true}$ or $F_{false}$ by the assignment;
	in each variable clause, update $F_1,F_2$ or $F_3$ by the variable satisfying the clause;
	consecutively update $F(k)$, for $k=\sqrt{a},\ldots,0$;
	complete the updates of $F_{true}$ or $F_{false}$;
	complete the updates of $F_1,F_2$ and $F_3$.
	We now detail these updates.
	
	In update step one, update
	in each variable gadget, $F_{true}$ or $F_{false}$ according to the assignment, in parallel. 
	For a variable $x_j$, this forms a $(\times(2-\epsilon))$-valid flow on $(w_j^1,w_j(0),w_j^2)$:
	the capacity of this path is $2a$, the increased capacity is $2a + a =3a \leq 4a-2\epsilon a$, as the initial capacity on it has demand $2a$, and the updated $F_{true}$ or $F_{false}$ has demand $a$.
	
	In step two, update
	in each clause gadget $C^i$, the flow $F_1,F_2$ or $F_3$ that corresponds to the variable $x_j$ that satisfies this clause.
	The updated flow is legal, as the edges $(u^i,u^i_j)$ and $(v^i_j,v^i)$ have enough capacity ($a$) and no flow on them, and the edge $(u^i_j,v^i_j)$ has capacity $a$ as well, and 
	after update step one, it has no flow on it. 
	If there is more than one variable satisfying $C^i$, we can choose any non-empty set of such variables.
	
	Next, update the blocking flow $F(\sqrt a)$.
	Each of the edges $(u^i,v^i)$ has capacity $3a$, increased capacity $(2-\epsilon)3a$, and after update step two, there are at most two flows of demand $a$ each on it, making this update valid, as $5a < (2-\epsilon)3a$.
	
	In the next $\sqrt a$ update rounds update the flows $F(k)$ for $k=\sqrt{a}-1,\ldots,0$, one per round. We will show below that these flows cannot be updated in the same round due to capacity constraints.
	These updates do not affect the flows between the gadgets.
	Inside each clause gadget, the capacity of the path $w_j^1,w(k+1),w_j^2$ is always greater than the flow $F(k)$, and this path is always unused when we update $F(k)$, since this update comes after the update of $F(k+1)$.
	
	In update round $\sqrt a + 4$ update the flows $F_{true}$ and $F_{false}$ in each variable gadget where it was not updated in the first round. 
	This is now valid since in the last round we  updated $F(0)$, so $(w_j^1,w(0),w_j^2)$ has capacity $2a$ and no flow on it.
	
	Finally, in the last update round we update the flows $F_1,F_2$ and $F_3$ in the clauses where they are not updated yet.
	This is possible as the last  update round made the edges of the form $(u^j_i,v^j_i)$ free of flow.
	This shows the existence of a $(\times(2-\epsilon))$-valid update schedule if a satisfying assignment exists. 
	
	For the other direction, assume the graph has a $(\times(2-\epsilon))$-valid update schedule.
	The construction of the flow pairs and the edge capacities implies that the update sequence we define above is the only one possible, which implies a satisfying assignment.
	We now prove this claim.
	
	We start by showing that only the flows $F_{true}$ and $F_{false}$ can be updated in the initial configuration, and only one of them in each variable gadget.
	Consider the variable gadget for $x_j$:
	Each path $(w^1_j,w_j(k),w^2_j)$ has capacity and flow $2a+k\sqrt a$; its excess capacity is thus
	\begin{equation*} 
		\begin{split}
			(1-\epsilon)(2a+k\sqrt a)
			&<2a+k\sqrt a-2\epsilon a\\
			&=2a+(k-1)\sqrt a +\sqrt a(1-2\epsilon\sqrt a)\\
			&<2a+(k-1)\sqrt a
		\end{split}
	\end{equation*}	
	where the last inequality uses the choice of $a$.
	Hence,
	for~all~$1\leq k\leq\sqrt a$, the flow $F(k-1)$ cannot be updated.
	For the flows $F_{true}$ and $F_{false}$, their updated path 
	$(w^1_j,w_j(0),w^2_j)$ has excess capacity $(1-\epsilon)2a<2a$, so they cannot both be updated simultaneously.
	
	Consider the clause gadget for $C^i$.
	The flows $F_1,F_2$ and $F_3$ have a demand of $a$ each, and their updated paths are of the form $(u^i,u^i_j,v^i_j,v^i)$; the central edge $(u^i_j,v^i_j)$ is saturated, and has excess capacity $(1-\epsilon)a$, smaller than the demand of $F_1$, $F_2$, and $F_3$. 
	Finally, the updated blocking flow $F(\sqrt a)$ uses edges of the form $(u^i,v^i)$; each such edge has capacity $3a$, $3$ flows of demand $a$ each, and its excess capacity $(1-\epsilon)3a$ is not enough to accommodate the $3a$ demand of~$F(\sqrt a)$.
	
	By this discussion, we see that the first update round can contain at most one $F_{true}$ or $F_{false}$ update in each variable gadget, and only these updates.
	Moreover, before any other type of update is made, there must be enough such updates to guarantee that all the clauses are satisfied: if this is not the case, there is a clause $C^i$ such that all its three edges of the form $(u^i_j,v^i_j)$ cannot take any of the flows $F_1,F_2$ and $F_3$.
	This, in turn, implies that the blocking flow $F(\sqrt a)$ can never be updated.
	Hence, before the blocking flow is updated, one of  $F_{true}$ and $F_{false}$ must be updated in each variable gadget, and then one of $F_1,F_2$ and $F_3$ must be updated in each clause gadget.
	So, a $(\times(2-\epsilon))$-valid update schedule implies as satisfying assignment, as claimed.
\end{proof}

\begin{figure}
 	\centering
 	\includegraphics[scale=.75]{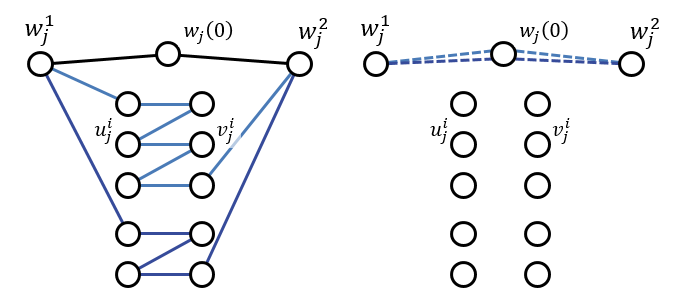}
 	\caption{The variable gadget for variable $x_j$ in the additive case.}
 	\label{fig: nph of additive bit gadget}
 \end{figure}

\adi
\begin{proof}
Given a 3-CNF formula, consider a graph construction as above, with $a=0$.
; see Figure~\ref{fig: nph of additive bit gadget}.
That is, for each $x_j$ the only node of the form $w_j(k)$ is $w_j(0)$, and the blocking flow is $F(0)=F(\sqrt a)$.
For $F_{true},F_{false}$, their updated flows in an $x_j$  gadget are set to be through the path $(w_j^1,w(0),w_j^2)$.

We now have only six flow pairs $F_{true},F_{false},F_1,F_2,F_3$ and $F(0)$, and set all their demands to $1$.
All the edges in the gadgets have capacity $1$, except for the edges of the paths $(w_j^1,w(0),w_j^2)$ with capacity $2$, and the edges of the form $(u^i,v^i)$ with capacity $3$. Thus $C_{\max}= 3$.

For the edges connecting $s$ and $t$ to the gadgets, and between the gadgets, set the capacity to $3$. 
In addition, the flows $F_1,F_2,F_3$ traverse the clause gadget in increasing order of identifiers, while the updated blocking flow $F(0)$ traverses them in the opposite order.
This ensures that the edges between the gadgets never form a bottleneck or prevent updates.
Given a satisfying assignment, the construction of a valid (and thus, a $(+(C_{\max}/3-\epsilon))$-valid) update schedule is similar to the one in the proof of Theorem~\ref{thm:2-eps is nph}.
In the first update round,  in each variable gadget $x_j$, update either $F_{true}$ or $F_{false}$, by the assignment.
This is possible since the path $(w_j^1,w(0),w_j^2)$ has capacity $2$ but only a single flow through it.
In the second update round in each clause  $C^i$, update one of the flows $F_1,F_2,F_3$ which goes through the edge $(u^i_j,v^i_j)$
corresponding to the variable $x_j$ that satisfies the clause.
This is possible since the edge $(u^i_j,v^i_j)$ now carries no flow.
In the third update round update the blocking flow $F(0)$, which uses the edges $(u^i,v^i)$ that were just now cleared out from one of the flows on them. Thus it has flow 2 and excess capacity $2 - \epsilon$, which is sufficient for the demand of flow $F(0)$.
In the forth update round, complete the update of $F_{true}$ and $F_{false}$, and in the fifth update round complete the update of $F_1,F_2$ and $F_3$, all without any additional augmentation in the edges of the graph.
For the converse direction, note that $C_{\max}/3-\epsilon=1-\epsilon$.
In the initial configuration, paths of the form $(w_j^1,w(0),w_j^2)$ have excess capacity of $1+(1-\epsilon)<2$, so only one of $F_{true}$ and $F_{false}$ can be updated before $F(0)$.
The edges $(u^i_j,v^i_j)$ have just $1-\epsilon$ excess capacity, so none of $F_1,F_2$ and $F_3$ can be updated in a clause gadget in which none of the variables was updated.
Hence, a $(+(C_{\max}/3-\epsilon))$-valid update schedule must first update one of $F_{true}$ and $F_{false}$ in each variable clause, in a way that induces a satisfying assignment, as claimed.
\end{proof}

\end{document}